%% file: letter.tex
\documentclass[aps,prl,twocolumn,floatfix,nofootinbib]{revtex4-2}

\usepackage{amsmath,amsthm,amssymb}
\usepackage{graphicx} 
\usepackage{mathtools}
\usepackage{physics}
\usepackage{bbold} 
\usepackage[colorlinks = true, allcolors = purple]{hyperref}
\usepackage{tikz}
\usetikzlibrary{arrows.meta, shapes.geometric, calc}

\definecolor{topo}{HTML}{1B6E6B}
\definecolor{topobar}{HTML}{9DD8C9}
\definecolor{triv}{HTML}{B45626}
\definecolor{trivbar}{HTML}{F4C2A6}
\definecolor{cor}{HTML}{4A3F9E}
\definecolor{ccgray}{HTML}{888888}
\definecolor{crit}{HTML}{000000}
\definecolor{avcol}{HTML}{A03C2C}
\definecolor{bpcol}{HTML}{2F5586}
\definecolor{xlcol}{HTML}{2F5586}
\definecolor{zlcol}{HTML}{A03C2C}

\theoremstyle{definition}
\newtheorem{theorem}{Theorem}
\newtheorem{lemma}{Lemma}

\newtheorem{defn}{Definition}

\newcommand{\term}[1]{\left( #1 \right)}
\renewcommand{\comm}[1]{\left[ #1 \right]}
\newcommand{\curly}[1]{\left\{ #1 \right\}}

\newcommand{\figref}[1]{Fig.~\ref{#1}}
\newcommand{\sfigref}[2]{Fig.~\hyperref[#1]{\ref{#1}#2}}

\usepackage{xcolor}
\definecolor{mydarkred}{RGB}{100,0,0}
\definecolor{mydarkblue}{RGB}{0,0,100}
\definecolor{mydarkgreen}{RGB}{30,120,30}

\DeclareMathOperator{\diam}{diam}
\DeclareMathOperator{\artanh}{artanh}

\newcommand{\CC}{\mathcal{C}}
\newcommand{\LL}{\mathcal{L}}

\newcommand{\TC}{\text{TC}}
\newcommand{\cL}{\mathcal{L}}
\newcommand{\Cclus}{C_{\mathrm{clus}}}

\newcommand{\Caltech}{Department of Physics and Institute for Quantum Information and Matter, California Institute of Technology, Pasadena, CA 91125, USA}
\begin{document}
\title{Gapped Parent Hamiltonians for the Strongly Deformed Toric Code}
\author{Nandagopal Manoj}\email{nmanoj@caltech.edu}
\author{Zack Weinstein}
\author{Jason Alicea}
\affiliation{\Caltech}
\date{\today}

\begin{abstract}
Local non-unitary deformations of topologically ordered wavefunctions can drive transitions into peculiar states that challenge modern perspectives on gapped quantum matter.
The strongly deformed toric code offers a curious case, hosting $m$ anyon condensation alongside perimeter-law scaling of Wilson loops charged under an exact 1-form symmetry---properties that typically do not coexist in gapped ground states. Nevertheless, we rigorously construct local gapped parent Hamiltonians for these strongly deformed toric code states. The Hamiltonians we construct are not strictly finite-range, but contain sums of Wilson loop operators whose coefficients decay exponentially in their diameter. If one adopts standard locality bounds used to define gapped phases---which allow for such exponentially decaying terms---our construction shows that these states realize a trivial gapped phase. Within this locality class, we demonstrate that perimeter-law scaling of Wilson loops does not imply a spontaneously broken 1-form symmetry, and from a dual perspective, that long-range ferromagnetic order and perimeter-law disorder parameter correlations can coexist in a 2D gapped ground state. We evade a recent no-go theorem [Sahay et al., arXiv:2503.01977] by relaxing its assumptions in a manner that we quantify as benign in the thermodynamic limit. More broadly, our results highlight that stronger notions of locality are necessary for prohibiting these counterintuitive properties within a gapped phase.
\end{abstract}

\maketitle

Ground states of local gapped Hamiltonians have been rigorously proven to exhibit exponentially decaying correlations~\cite{HastingsKoma2006,hastings_locality_2010}. The nature of the converse relationship---the sufficient conditions for a quantum state to be the gapped ground state of a local Hamiltonian---remains a fundamental open problem outside of one spatial dimension. In one dimension, the identification of matrix product states as the relevant `corner of Hilbert space' enables an analytical classification of gapped phases~\cite{Haldane1983a,Haldane1983b,Chen_Gu_Wen_2010,Chen_Gu_Wen_2011,SchuchGarciaCirac2011} and underlies powerful variational techniques for their numerical study \cite{white_density_1992,ostlund_thermodynamic_1995,vidal_efficient_2003,verstaete_matrix_2006,cirac_matrix_2021}. Attaining an analogous understanding in two (and higher) dimensions potentially paves the way for similar applications.

Deformed toric code states [see Eq.~\eqref{eqn:defnofdefTC}] provide an interesting edge case that subtly challenges modern methods of classifying phases. These states arise from deforming toric code ground states by a non-unitary operator that creates $m$-anyon pairs, driving a quantum phase transition out of the topological phase beyond a critical deformation strength \cite{CastelnovoChamon2007DefToricCode}. In the resulting strongly deformed phase, a peculiar combination of properties emerges: (i) condensed $m$ anyons; (ii) an exact `electric' 1-form symmetry; and (iii) \emph{perimeter-law} scaling of Wilson loops \cite{HuxfordNguyenKim2023}. The final property is especially puzzling, given the first two: such a perimeter law is typically interpreted as a signature of deconfined $e$ anyons and a spontaneously broken 1-form symmetry~\cite{Gaiotto2015, McGreevyGLP2023}, both of which appear incompatible with a conventional $m$-condensed gapped phase~\cite{Burnell_anyon_2018}. Based on this observation, Ref.~\onlinecite{Sahay2025FunkyIsing} proved a rigorous no-go theorem showing---subject to several technical assumptions---that the strongly deformed toric code states are \emph{intrinsically gapless}; that is, they do not admit a local gapped parent Hamiltonian with a finite ground-state degeneracy.

\begin{figure}
\centering
\includegraphics[width = \columnwidth]{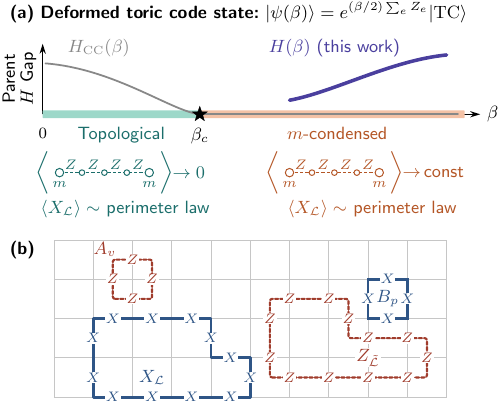}
\caption{(a) Properties of the deformed toric code state $\ket{\psi(\beta)}$.
For $\beta$ beyond a critical deformation strength $\beta_c$, $m$ anyons condense and destroy the topological order, yet Wilson loops continue to exhibit perimeter-law scaling. The vertical axis on the phase diagram sketches the spectral gap of two parent Hamiltonians for $\ket{\psi(\beta)}$: the Castelnovo-Chamon Hamiltonian $H_{\rm CC}$ [Eq.~\eqref{eq:HCC}] and our parent Hamiltonian $H(\beta)$ [Eq.~\eqref{eqn:Hbeta}], proved to be gapped and
non-degenerate for sufficiently large $\beta$.
(b) Toric-code stabilizers $A_v$ and $B_p$, the associated $1$-form symmetry
$Z_{\tilde{\mathcal{L}}}$, and Wilson loop operator $X_{\mathcal{L}}$.}
\label{fig:overview}
\end{figure}

\begin{table}[t]
\caption{Corollaries of the parent Hamiltonian construction. Gapped ground state is in reference to an exponential-in-diameter-local Hamiltonian [cf. Eq.~(\ref{eq:locality_bound})].}
\label{tab:corollaries}
\begin{ruledtabular}
\begin{tabular}
{p{0.95\columnwidth}}

\hangindent=1.2em
\hangafter=1
1.\ Perimeter-law Wilson loops in a gapped ground state do not imply $1$-form SSB. \\

\hangindent=1.2em
\hangafter=1
2.\ Long-range order parameter correlations and perimeter-law disorder correlations can coexist in a 2D gapped ground state. \\

\hangindent=1.2em
\hangafter=1
3.\ Toric code under strong Pauli-$Z$ decoherence is separable as a convex sum of trivial gapped ground states.   \\

\hangindent=1.2em
\hangafter=1
4.\ 2D classical Ising ferromagnet at low temperature is dual to a $\mathbb{Z}_2$ gauge theory with exponentially decaying interactions at high temperature. \\
\end{tabular}
\end{ruledtabular}
\end{table}

In this Letter, we explicitly construct local parent Hamiltonians for which the strongly deformed toric code states are the unique gapped ground states. We circumvent the no-go theorem of Ref.~\onlinecite{Sahay2025FunkyIsing} by relaxing its assumptions in a physically motivated way---such that the violations to these assumptions vanish exponentially quickly in the thermodynamic limit. {Our parent Hamiltonians naturally incorporate the perimeter law by containing sums of Wilson loop operators with coefficients decaying exponentially in their support; although the resulting Hamiltonians are not strictly finite-range, they nevertheless satisfy standard locality bounds commonly used to define and analyze gapped phases of matter~\cite{HastingsKoma2006,hastings_locality_2010}.} More broadly, we demonstrate that within the standard definition of locality, perimeter-law decay of Wilson loops is compatible with a trivial symmetric phase and is therefore not a reliable probe of 1-form spontaneous symmetry breaking (SSB). Our construction yields a number of nontrivial corollaries, highlighted in Table \ref{tab:corollaries}.

{\bf \emph{Deformed toric code.}}~Consider an $L_x \times L_y$ square lattice with periodic boundary conditions, where each edge $e$ hosts a qubit equipped with Pauli operators $X_e,Z_e$. The code space of the toric code is the simultaneous +1 eigenspace of the vertex stabilizers $A_v = \prod_{e \ni v} Z_e$ and the plaquette stabilizers $B_p = \prod_{e \in p} X_e$ [see Fig.~\ref{fig:overview}(b)]. In other words, it is the four-dimensional ground-state subspace of the Hamiltonian \cite{kitaev_fault-tolerant_2003}
\begin{equation} \label{eqn:H0TC}
    H_{\TC} = \sum_v(1-A_v) +\sum_p(1- B_p).
\end{equation}
We will initially focus on the logical `++' toric code state $\ket{\TC} = \sum_{\mathcal{L}} X_{\mathcal{L}} \ket{0}$, where $X_{\mathcal{L}} = \prod_{e \in \mathcal{L}} X_e$ is a product over a collection of (possibly disconnected or non-contractible) closed loops $\mathcal{L}$ in the direct lattice, and $\ket{0}$ is the simultaneous +1 eigenstate of each $Z_e$. This state satisfies $X_{\mathcal{L}} \ket{\TC} =  \ket{\TC}$ for any loop $\mathcal{L}$, and is therefore the +1 eigenstate of the two logical-$X$ operators.

Given $\ket{\TC}$, the (unnormalized) deformed toric code state is defined by \cite{CastelnovoChamon2007DefToricCode}
\begin{align} \label{eqn:defnofdefTC}
    \ket{\psi(\beta)} &=  e^{\frac{\beta}{2}\sum_e Z_e}\ket{\TC} \propto 
    \sum_{\cL} e^{-\beta \abs{\cL}} X_{\cL}\ket{0}
\end{align}
with $\abs{\cL}$ the number of edges in $\cL$. This state is symmetric under the 1-form symmetry $Z_{\tilde{\mathcal{L}}}$ [Fig.~\ref{fig:overview}(b)] for any \emph{contractible} loop $\tilde{\mathcal{L}}$ in the dual lattice, and thus differs slightly from the deformed state considered in Ref.~\onlinecite{Sahay2025FunkyIsing}; we address the latter as well in later sections. 

The norm of $\ket{\psi(\beta)}$, 
\begin{align} \label{eqn:partition_function}
    \braket{\psi(\beta)}{\psi(\beta)} \propto \sum_{\cL} e^{-2 \beta \abs{\cL}},
\end{align}
is proportional to the partition function of the two-dimensional (2D) classical Ising model, where each $\mathcal{L}$ represents a configuration of magnetic domain walls~\cite{fn1}. This connection suggests that $\ket{\psi(\beta)}$ exhibits a phase transition in the 2D classical Ising$^\star$ universality class at a critical deformation strength $\beta_c = \frac{1}{2}\ln(1 + \sqrt{2}) \approx 0.441$. 
Within the strongly deformed phase $\beta > \beta_c$, $\ket{\psi(\beta)}$ exhibits long-range order in the string operator $Z_{\tilde P} =\prod_{e \in {\tilde P}} Z_e$ that creates $m$ anyons at the endpoints of the dual-lattice path ${\tilde P}$,
indicating that $m$ anyons are condensed~\cite{Burnell_anyon_2018}. The resulting state is topologically trivial, as indicated by a vanishing topological entanglement entropy~\cite{CastelnovoChamon2007DefToricCode, LevinWen2006, KitaevPreskill2006}. 

In typical Hamiltonian deformations of $H_{\TC}$ that condense the $m$ anyons while preserving the exact 1-form symmetry \cite{Wegner_duality_1971,KogutRMP1979}, the resulting ground states exhibit area-law scaling of Wilson loop expectation values: $\expval{X_{\mathcal{L}}} \sim e^{-\alpha \text{Area}(\mathcal{L})}$, for a large contractible loop $\mathcal{L}$. In contrast, Ref.~\onlinecite{HuxfordNguyenKim2023} emphasized that $\ket{\psi(\beta)}$ exhibits perimeter-law scaling $\expval{X_{\mathcal{L}}} \sim e^{-\mu \abs{\mathcal{L}}}$ for any finite $\beta$, even within the strongly deformed phase. These authors interpreted the perimeter-law scaling of Wilson loops in the $\beta > \beta_c$ phase as indicative of spontaneous 1-form symmetry-breaking \emph{without} topological order, in apparent violation of the generalized Landau paradigm~\cite{McGreevyGLP2023}. Subsequently, Ref.~\onlinecite{Sahay2025FunkyIsing} suggested that the combination of exact 1-form symmetry under $Z_{\tilde{\mathcal{L}}}$, long-range order in $\expval{Z_{\tilde P}}$, and perimeter-law scaling in $\expval{X_{\mathcal{L}}}$ cannot arise in the ground state of a gapped Hamiltonian with a finite degeneracy. Here we will offer an alternative perspective on the state $\ket{\psi(\beta)}$ by constructing, at large deformation strengths $\beta$, an exact local parent Hamiltonian for which $\ket{\psi(\beta)}$ is the unique gapped ground state.

{\bf \emph{Parent Hamiltonian constructions.}}~Equation~\eqref{eqn:defnofdefTC} reveals two complementary perspectives on $\ket{\psi(\beta)}$. At small $\beta$, where the topological order persists, $\ket{\psi(\beta)}$ is most naturally viewed as a deformation of the $H_{\rm TC}$ ground state $\ket{\rm{TC}}$.  Conversely, large $\beta$ massacres the topological order and exponentially suppresses large loops in the representation on the right side of Eq.~\eqref{eqn:defnofdefTC}.  There it is more natural to view $\ket{\psi(\beta)}$ as a deformation of the trivial product state $\ket{0}$, which is the unique zero-energy ground state of
\begin{equation}\label{eqn:H0}
  H_0 = \sum_e (1-Z_e).
\end{equation}
We will show that these two viewpoints lead to distinct parent Hamiltonians---adiabatically connected to $H_{\TC}$ and $H_0$, respectively---that are local and gapped in their `natural' regimes.  

Let us start with the former viewpoint and define manifestly
Hermitian, positive-semidefinite operators $Q_p$ by the following deformation of the $B_p$ projectors in Eq.~\eqref{eqn:H0TC}:
\begin{equation}
    \begin{split}
        Q_p &= e^{-\frac{\beta}{2} \sum_{e \in p} Z_e} (1 - B_p) e^{- \frac{\beta}{2} \sum_{e \in p} Z_e} \\
            &= e^{-\beta \sum_{e \in p} Z_e} - B_p.
    \end{split}
\end{equation}
Notice that the exponentials in the first line contain only pieces of the non-unitary deformation $e^{ \frac{\beta}{2} \sum_e Z_e}$ from Eq.~\eqref{eqn:defnofdefTC} that do not commute with a particular $B_p$. All of the operators $(1-A_v)$ and $Q_p$ clearly annihilate $\ket{\psi(\beta)}$, which accordingly admits the finite-range, frustration-free parent Hamiltonian
\begin{equation}
    H_{\rm CC}(\beta) = \sum_v(1-A_v) + \sum_p Q_p
    \label{eq:HCC}
\end{equation}
for any $\beta$. One can readily show that $H_{\text{CC}}(\beta)$ is an exact parent Hamiltonian for any toric code ground state deformed by $e^{\frac{\beta}{2}\sum_e Z_e}$.  
Equation~\eqref{eq:HCC} is precisely the Castelnovo-Chamon Hamiltonian \cite{CastelnovoChamon2007DefToricCode} and is indeed gapped for $\beta<\beta_c$ but, interestingly, gapless for $\beta > \beta_c$~\cite{Sahay2025FunkyIsing}.

Next we apply similar logic at large $\beta$, where $\ket{\psi(\beta)}$ appears `close' to $\ket{0}$. First note that the operator 
\begin{equation}\label{eq:Vdef1}
\mathcal{Z} = \sum_{\cL} e^{-\beta \abs{\cL}} X_{\cL} \propto \sum_{\qty{s_v = \pm 1}} e^{ K \sum_{\expval*{v v'}} X_{vv'} s_v s_{v'}  }
\end{equation}
acting on $\ket{0}$ in Eq.~\eqref{eqn:defnofdefTC} is positive-definite for all $\beta > 0$. This observation follows from the right side of Eq.~\eqref{eq:Vdef1}, which identifies $\mathcal{Z}$ with the partition function of a bond-disordered Ising model at inverse temperature $K = \artanh(e^{-\beta}) >0$.  There, $s_v$'s represent auxiliary classical spins on the vertices $v$ of the lattice, and the Pauli operators $X_{v v'} \equiv X_e$ on the links $e$ connecting vertices $v,v'$ encode signs of the Ising couplings. Positive-definiteness of $\mathcal{Z}$ guarantees that the operator
\begin{equation}
    W = \log \mathcal{Z} \equiv \sum_{\cL} J_{\cL} X_{\cL}
    \label{eq:Vdef}
\end{equation}
is well-defined and Hermitian. Furthermore, $W$ commutes with each vertex stabilizer $A_v$ and hence can be expanded as a sum over loop operators $X_{\mathcal{L}}$ with real-valued coefficients $J_{\mathcal{L}}$, as shown above.

At this point we have expressed $\ket{\psi(\beta)}$ as the trivial state $\ket{0}$ deformed by a non-unitary operator $\mathcal{Z} = e^{W}$, where $W$ is a sum of commuting operators---paralleling the middle representation in Eq.~\eqref{eqn:defnofdefTC}.  Defining operators $W_e = \sum_{\mathcal{L} \ni e} J_{\mathcal{L}} X_{\mathcal{L}}$ that contain only the components of $W$ that anticommute with $Z_e$, we analogously obtain our alternate
frustration-free parent Hamiltonian for $\ket{\psi(\beta)}$,
\begin{equation} \label{eqn:Hbeta}
    \begin{split}
        H(\beta) &= \sum_e e^{-W_e}\left(1 - Z_e\right)e^{-W_e} \\
        &= \sum_e \qty( e^{- 2W_e} - Z_e ).
    \end{split}
\end{equation}
Indeed, similar to the Castelnovo-Chamon Hamiltonian~(\ref{eq:HCC}), $H(\beta)$ sums Hermitian, positive-semidefinite terms that individually annihilate $\ket{\psi(\beta)}$.

Contrary to the Castelnovo-Chamon Hamiltonian, locality of $H(\beta)$ is not manifest but can be argued heuristically at large $\beta$ as follows.  In the extreme limit $\beta \to \infty$, each $J_{\mathcal{L}}$ for non-empty loop configurations $\mathcal{L}$ vanishes, and $H(\beta \to \infty)$ reduces to the trivial local parent Hamiltonian \eqref{eqn:H0}. More generally, each $J_{\mathcal{L}}$ decays asymptotically in $\abs{\mathcal{L}}$ as $e^{-\beta \abs{\mathcal{L}}}$, exponentially suppressing large loop contributions to $W_e$. Furthermore, as we carefully explain below, the logarithm in Eq.~\eqref{eq:Vdef} ensures that potentially problematic nonlocal loop configurations consisting of far-separated disconnected loops are effectively absent in $W$.
This observation---which one can readily verify to low orders by Taylor expanding the logarithm---closely relates to the linked-cluster theorem from diagrammatic perturbation theory \cite{zinn-justin_quantum_2021}, and crucially ensures extensivity of the ``free energy'' $W$ for the disordered Ising model $\mathcal{Z}$ \cite{kardarStatisticalPhysicsFields2007}. Consequently, it is natural to expect that each operator $e^{-2W_e}$ is locally supported near the edge $e$, and can be organized as a sum of terms that decay exponentially in their diameter about $e$ in the large-$\beta$ regime. 

{\bf \emph{Proving locality and spectral gap.}}~We now outline a rigorous proof that, for sufficiently large $\beta$, our parent Hamiltonian $H(\beta)$ is local in the sense of Refs.~\onlinecite{HastingsKoma2006,hastings_locality_2010, Sahay2025FunkyIsing} and consequently exhibits a finite spectral gap. Technical details are deferred to the Supplemental Material~\cite{SOM}. Our approach uses the Mayer cluster expansion~\cite{Friedli_Velenik_2017} to explicitly write Eq.~\eqref{eq:Vdef} as a sum of exponentially localized terms, bolstering the intuition developed above. Locality of $H(\beta)$ arises directly from properties of the cluster expansion, and the spectral gap at large $\beta$ then follows from the gap stability results of Refs.~\onlinecite{BravyiHastingsMichalakis2010,BravyiHastings2011}.

To this end, we interpret $\mathcal{Z}$ in Eq.~(\ref{eq:Vdef1}) as the (operator-valued) partition function for a hard-core \emph{polymer model}. Each loop configuration $\mathcal{L}$ is first decomposed into a collection of connected loops $\gamma$ (`polymers'), each carrying the operator-valued weight $w(\gamma) = e^{-\beta \abs{\gamma}} X_{\gamma}$. Each pair of polymers $\gamma, \gamma'$ is then assigned a hard-core interaction $\delta(\gamma, \gamma') = 0$ if they intersect and $1$ otherwise. With this notation, Eq.~\eqref{eq:Vdef1} can be rewritten as
\begin{equation}
    \mathcal{Z} = \sum_{\Gamma' \subseteq \Gamma } \left[\prod_{\gamma \in \Gamma'} w(\gamma) \right] \left[\prod_{\curly{\gamma, \gamma'} \subseteq \Gamma'}  \delta(\gamma, \gamma')\right] ,
    \label{eqn:polymer_model}
\end{equation}
where $\Gamma$ denotes the set of all possible polymers $\gamma$, and $\Gamma'$ sums over all subsets of $\Gamma$. The hard-core interactions ensure that the contributing subsets $\Gamma'$ are in one-to-one correspondence with the loop configurations $\mathcal{L}$ in Eq.~\eqref{eq:Vdef1}.

The Mayer cluster expansion expresses the `polymer free energy' $W$ as a sum over \emph{connected clusters} $\mathcal{C}$, each defined as a (possibly repeating) finite sequence of polymers $\term{\gamma_{1}, \gamma_{2}, \ldots, \gamma_{n}}$ whose union is connected~\cite{Friedli_Velenik_2017}.  
Eq.~\eqref{eq:Vdef} accordingly admits the more constrained form
\begin{equation}
    W = {\sum_{\mathcal{C}} f(\mathcal{C}) X_{\mathcal{C}}} ,
    \label{eq:V_Mayer}
\end{equation}
where $X_{\mathcal{C} = (\gamma_1,\dots,\gamma_n)} \equiv  X_{\gamma_1}X_{\gamma_2} \dots X_{\gamma_n}$, and the coefficients $f(\mathcal{C})$ are defined precisely in the End Matter.  We can then identify $W_e  = \sum_{\mathcal{C} \in \mathcal{S}_e} f(\CC) X_{\CC}$, where $\mathcal{S}_e \equiv \curly{\CC: \text{support}(X_{\CC}) \ni e}$ defines the set of connected clusters for which $X_\mathcal{C}$ anticommutes with $Z_e$. Our parent Hamiltonian then admits the useful representation $H(\beta) = H_0 + V(\beta)$, where
\begin{equation}
    V(\beta) = \sum_e \term{\exp\comm{-2\sum_{\mathcal{C} \in \mathcal{S}_e} f(\CC) X_{\CC}}-1}.
    \label{eq:HbetafC}
\end{equation}

In the Supplemental Material \cite{SOM}, we derive several properties of the coefficients $f(\mathcal{C})$ relevant for the locality of $H(\beta)$. We first show that the coefficients decay exponentially with the size of the constituent polymers defining $\mathcal{C}$---i.e., $\abs{f(\mathcal{C})} \lesssim \prod_{\gamma \in \mathcal{C}} e^{-\beta \abs{\gamma}}$. This fact alone does not suffice for locality, since each $W_e$ sums over combinatorially many connected clusters. To prove that $H(\beta)$ is bounded and local, we further derive the crucial inequality (see also Ref. \onlinecite{FernandezProcacci2007} for a tighter bound)
\begin{equation} \label{eqn:fCbound}
    \sum_{\mathcal{C} \in \mathcal{S}_e} \abs{f(\CC)} \leq \mathrm{const} \times e^{-4\beta}
\end{equation}
which involves the sum of \emph{all} $f(\mathcal{C})$'s contributing to a given $W_e$, and holds for $\beta > \beta^* = 2 + \log 3$.  These bounds are proved using graph-theoretic techniques that are standard to the cluster expansion~\cite{Friedli_Velenik_2017,Penrose1967,KoteckyPreiss1986}; the nontrivial step involves bounding the $\sim n!$ growth in the number of ways the same set of $n$ polymers can form a cluster. 
Equation~\eqref{eqn:fCbound}, together with the exponential decay of $f(\CC)$ in cluster size, ensures that  $H(\beta)$ is a sum of exponentially localized operators for sufficiently large $\beta$, in harmony with the intuition developed earlier.

Finally, since $H(\beta)$ in Eq.~\eqref{eqn:Hbeta} is obtained from a local perturbation $V(\beta)$ to the trivial gapped Hamiltonian $H_{0}$ in Eq.~\eqref{eqn:H0}, we can leverage the results of Bravyi et al.~\cite{BravyiHastingsMichalakis2010,BravyiHastings2011} to prove that $H(\beta)$ is also gapped for sufficiently large $\beta$. Explicitly, we first decompose $V(\beta) = \sum_{r \geq 1} \sum_{A \in S(r)} V_{r,A}$ into a sum of local terms, where $S(r)$ is the set of all $r \times r$ squares $A$ on the lattice, and $V_{r,A}$ exhibits trivial support outside of $A$. Next, we prove in the Supplemental Material~\cite{SOM} that the operator norm $\norm{V_{r,A}}$ of each local term can be (conservatively) upper-bounded as
\begin{equation}
\label{eq:VrA}
\norm{V_{r,A}} \leq J(\beta) e^{-r}, \quad J(\beta) = e^{-4 \beta} \times (4.11 \times 10^5) ,
\end{equation}
which shows that $\norm{V_{r,A}}$ decays (at least) exponentially in the diameter of $A$. With this precise definition of a local perturbation $V(\beta)$, Ref.~\onlinecite{BravyiHastingsMichalakis2010} proved that the gap $\Delta(\beta)$ of $H(\beta)$ is lower-bounded as $\Delta(\beta) \geq [1 - c_{1} J(\beta)] \Delta_{0}$ up to corrections that vanish super-polynomially in system size~\cite{fn2}, where $c_{1}$ is a constant and $\Delta_{0} = 2$ is the gap of $H_{0}$. Since $J(\beta)$ decays exponentially with $\beta$, we have shown that for sufficiently large $\beta$, our local parent Hamiltonian $H(\beta)$ admits the unique ground state $\ket{\psi(\beta)}$ with a finite spectral gap $\Delta(\beta)$~\cite{fn3}.

\textbf{\textit{Consistency with the no-go theorem.}}~Let us describe the relationship between our construction and the rigorous no-go theorem from Ref.~\onlinecite{Sahay2025FunkyIsing}. Thus far we have focused on the deformed logical `++' state---Eq.~\eqref{eqn:defnofdefTC}---which we now denote more explicitly by $\ket{\psi_{++}(\beta)}$. Our parent Hamiltonian $H(\beta)$ for this state does not violate the no-go theorem, since $\ket{\psi_{++}(\beta)}$ does not satisfy all of its assumptions: namely, the theorem requires an exact 1-form symmetry $Z_{\tilde{\mathcal{L}}}$ about both contractible \emph{and} non-contractible loops $\tilde{\mathcal{L}}$ in the dual lattice. To meet these assumptions, we should instead consider the logical `00' state $\ket{\psi_{00}(\beta)} \propto \sum_{\mathcal{L}}' e^{-\beta \abs{\mathcal{L}}} X_{\mathcal{L}} \ket{0}$, where the prime indicates that the sum is performed over only contractible loops $\mathcal{L}$. 

Since non-contractible loops in $\ket{\psi_{++}(\beta)}$ are exponentially suppressed in the linear system size $L \equiv \min(L_{x}, L_{y})$ throughout the strongly deformed phase $\beta > \beta_{c}$, the (normalized) fidelity between $\ket{\psi_{00}(\beta)}$ and $\ket{\psi_{++}(\beta)}$ is lower-bounded by $1 - e^{-\mathcal{O}(L)}$ at large $\beta$. Consequently, these two states are indistinguishable from each other in the thermodynamic limit; in particular, $H(\beta)$ is an excellent \emph{approximate} parent Hamiltonian for $\ket{\psi_{00}(\beta)}$, with a variational energy $e^{-\mathcal{O}(L)}$. Since phases of matter are strictly well-defined in the thermodynamic limit $L \to \infty$, we should regard $\ket{\psi_{00}(\beta)}$ as belonging to the same phase of matter as $\ket{\psi_{++}(\beta)}$, namely, a trivial gapped phase. This observation is again consistent with the results of Ref.~\onlinecite{Sahay2025FunkyIsing}, which focuses on the existence of \emph{exact} parent Hamiltonians.

Nevertheless, by employing identical cluster expansion techniques as in the previous section, we can indeed construct an exact 1-form symmetric parent Hamiltonian $H^{(00)}(\beta)$ for which $\ket{\psi_{00}(\beta)}$ is the unique gapped ground state \cite{SOM}. This Hamiltonian differs from Eq.~\eqref{eq:HbetafC} by terms of total spectral norm $e^{-\mathcal{O}(L)}$ in the large-$\beta$ regime, and satisfies a locality criterion analogous to Eq.~\eqref{eq:VrA} for any \emph{fixed} aspect ratio $a_{xy} \equiv L_{x} / L_{y}$. Once more, this result is consistent with the no-go theorem, which forbids the existence of a parent Hamiltonian obeying an \emph{aspect-ratio-independent} locality bound. Specifically, Ref.~\onlinecite{Sahay2025FunkyIsing} begins from the standard definition of local Hamiltonians $H = \sum_{\mathcal{R} \subseteq \Lambda} h_{\mathcal{R}}$ with exponentially decaying tails \cite{HastingsKoma2006,hastings_locality_2010}, which are required to satisfy the bound \cite{fn4}
\begin{equation}
    \label{eq:locality_bound}
    \sup_{i \in \Lambda} \sum_{\mathcal{R} \ni i} \norm{h_{\mathcal{R}}} \abs{\mathcal{R}} e^{\mu \diam (\mathcal{R})} \leq s
\end{equation}
for finite constants $\mu, s$, but additionally demands that $\mu, s$ can be chosen independently of the aspect ratio $a_{xy}$. On the other hand, our Hamiltonian $H^{(00)}(\beta)$ contains terms which create \emph{pairs} of far-separated non-contractible closed loops. Although the total norm of such terms is exponentially suppressed in $L$, the aspect ratio can always be chosen sufficiently small that Eq.~\eqref{eq:locality_bound} is violated for any constants $\mu, s$ fixed in advance. Our parent Hamiltonian construction therefore demonstrates that the aspect-ratio-independent locality criterion is not just a technical assumption, but physically necessary for proving enforced gaplessness of $\ket{\psi_{00}(\beta)}$.

\textbf{\textit{Discussion.}}~
We have demonstrated that the strongly deformed toric code state $\ket{\psi(\beta)}$ admits a local and gapped parent Hamiltonian $H(\beta)$ with $\ket{\psi(\beta)}$ as the unique ground state. Since $H(\beta)$ is connected to the trivial Hamiltonian $H_0$ by a gap-preserving deformation, this result demonstrates that $\ket{\psi(\beta)}$ lies within the trivial phase for sufficiently large $\beta$. 
Although our cluster expansion techniques only prove that $H(\beta)$ is local and gapped at large $\beta$, it is natural to conjecture that the entire $\beta > \beta_c$ regime realizes a trivial phase with a corresponding gapped parent Hamiltonian.

An immediate corollary of our result is that, within the space of Hamiltonians exhibiting exponentially decaying interactions [see Eq.~\eqref{eq:locality_bound}], perimeter-law decay of Wilson loops is not a reliable diagnostic of 1-form SSB. In the presence of an exact 1-form symmetry, it is commonly argued \cite{Gaiotto2015,McGreevyGLP2023} that a perimeter-law Wilson loop can be ``locally dressed'' into an additional topological operator that braids nontrivially with the 1-form symmetry, resulting in topological order and ground-state degeneracy on the torus. Our construction demonstrates that this intuition fails in the deformed toric code states. A natural question is whether there might be more robust diagnostics of 1-form SSB at the level of individual states, i.e., without reference to a particular parent Hamiltonian. 

The techniques developed in this work also shed light on the mixed-state separability of the strongly \emph{decohered} toric code. Previously, Refs.~\onlinecite{ChenGrover_separability_2024,WangSongMengGrover_analog_2025} demonstrated that the toric code under strong Pauli-$Z$ decoherence can be expressed as an incoherent mixture of topologically trivial pure states with condensed $m$ anyons. These states are structurally similar to the deformed toric code states, and given the results of Ref.~\onlinecite{Sahay2025FunkyIsing}, one might wonder whether they are intrinsically gapless. As described in the End Matter, and elaborated in the Supplemental Material \cite{SOM}, one can construct exact gapped parent Hamiltonians for these states at large decoherence strengths by an identical procedure as for the deformed toric code. This result completes the argument of Refs.~\onlinecite{ChenGrover_separability_2024,WangSongMengGrover_analog_2025} into a rigorous proof that the strongly $Z$-decohered toric code is separable, i.e., can be written as a convex sum of short-range entangled pure states.

{From a dual perspective, our result also demonstrates that gapped Ising-symmetric wavefunctions can simultaneously exhibit long-range ferromagnetic correlations and perimeter-law disorder parameter correlations. Conventionally, when a gapped Hamiltonian exhibits a spontaneously broken global (0-form) Ising symmetry, its symmetric ground state exhibits long-range order parameter correlations and \emph{area-law} disorder parameter correlations. Indeed, a no-go theorem~\cite{Levin_constraints_2020} rules out the coexistence of long-range order parameter and long-range disorder parameter correlations in 1D, and a recent result~\cite{McDonough_Zhang_lieb-robinson_2025} demonstrates that \emph{exponential-in-volume} local perturbations to the 2D ferromagnetic Ising fixed point result in area-law disorder parameters. In contrast, the strongly deformed toric code state $\ket{\psi_{00}(\beta)}$ is dual to a ferromagnetic ground state $\ket{\varphi(\beta)}$ with perimeter-law disorder parameters, and our gapped parent Hamiltonian $H^{(00)}(\beta)$ is correspondingly dual to a gapped Ising-symmetric parent Hamiltonian $H_{\text{dual}}(\beta)$ for $\ket{\varphi(\beta)}$; see the End Matter for details. Interestingly, although $H_{\text{dual}}(\beta)$ satisfies the exponential-in-diameter locality bound of Eq.~\eqref{eq:locality_bound}, it \emph{fails} to satisfy the more stringent exponential-in-volume locality criterion. We expect that with a restriction to exponential-in-volume local Hamiltonians, long-range order parameters and area-law disorder parameters cannot coexist in a gapped ground state.}

{Our results highlight that the specific notion of locality imposed on the space of allowed Hamiltonians is a subtle but physically consequential choice in the classification of gapped quantum phases. It is particularly interesting to consider whether a more restrictive notion of locality can resurrect the correspondence between perimeter-law Wilson loops, 1-form SSB, and topological order. Crucially, the exponential-in-volume locality criterion previously discussed is insufficient for this purpose: in the Supplemental Material~\cite{SOM}, we show that our parent Hamiltonians $H(\beta)$ and $H^{(00)}(\beta)$ satisfy even these locality bounds. Inspired by the dual Ising parent Hamiltonian, perhaps the physically correct notion of locality in 1-form symmetric systems requires the suppression of Wilson loop operators $X_{\mathcal{L}}$ not by the volume $\abs{\mathcal{L}}$ of their support, but by the area of the region they bound. Alternatively, it may be preferable to focus on Hamiltonian-agnostic criteria for gapped phases, such as the (approximate) entanglement bootstrap axioms~\cite{Shi_fusion_2020,Shi_Kim_2021,Kim_Lin_Ranard_Shi_2024}. Given that the deformed toric code states satisfy these axioms as well~\cite{SOM}, it is interesting to ask whether there might exist a stronger Hamiltonian-agnostic definition of a gapped state which restores the intrinsic gaplessness of these funky states.}

\textit{Note added:} During the completion of this work, we became aware of a forthcoming work by Sahay, Zhang, von Keyserlingk, and Verresen that overlaps in part with the mixed state results of this work. Where our results overlap they agree.

\textbf{\textit{Acknowledgements.}}~
We thank Rahul Sahay for inspiring discussions and introducing us to this problem. We also thank Tarun Grover, Ethan Lake, Anton Kapustin, John McGreevy, Spyridon Michalakis, Olexei Motrunich, Zohar Nussinov, Daniel Ranard, Ruben Verresen, Curt von Keyserlingk, and Carolyn Zhang for helpful discussions and comments. This work was primarily supported by the U.S. Department of Energy, Office of Science, National Quantum Information Science Research Centers, Quantum Science Center. We also acknowledge funding provided by the Institute for Quantum Information and Matter, an NSF Physics Frontiers Center (NSF Grant PHY-2317110).

\makeatletter
\let\old@addcontentsline\addcontentsline
\renewcommand{\addcontentsline}[3]{}
\makeatother

\bibliography{biblio}
\addcontentsline{toc}{section}{}

\clearpage
\newpage
\onecolumngrid
\section{End Matter}
\twocolumngrid

\textbf{\textit{Exact expression for the parent Hamiltonian.}}~We now provide exact expressions for the coefficients $f(\mathcal{C})$ appearing in Eq.~\eqref{eq:HbetafC}, thereby fully specifying the parent Hamiltonian $H(\beta)$ for the strongly deformed toric code $\ket{\psi(\beta)} = e^{W} \ket{0}$ at large $\beta$.  In our derivation we wrote 
\begin{equation}
     W = \log \mathcal{Z} = \sum_\mathcal{C} f(\mathcal{C})X_\mathcal{C} ,
     \label{eq:W_EM}
\end{equation} 
where $\mathcal{Z}$ is the polymer-model partition function in Eq.~\eqref{eqn:polymer_model}. The weights and hard-core interaction parameters in $\mathcal{Z}$ respectively read
\begin{equation}
    w(\gamma) = e^{-\beta \abs{\gamma}} X_\gamma, \,\,\, \delta(\gamma, \gamma') = \begin{cases}
        0 & \gamma, \gamma'\text{ intersect} \\
        1 & \text{otherwise}
    \end{cases},
\end{equation}
where $X_{\gamma} = \prod_{e \in \gamma} X_e$, and two polymers $\gamma, \gamma'$ intersect if there exists edges $e \in \gamma$ and $e' \in \gamma'$ such that $e, e'$ share a vertex on the square lattice.  (For example, this convention counts a `figure 8' as a single polymer, rather than two intersecting polymers.) 

The Mayer cluster expansion~\cite{Friedli_Velenik_2017} rewrites the polymer free energy $W$ as a sum over connected clusters $\mathcal{C}$, allowing us to extract $f(\mathcal{C})$ through the right side of Eq.~\eqref{eq:W_EM}. Specifically,
the cluster expansion provides the decomposition
\begin{equation}
    \log \mathcal{Z} = \sum_{n = 1}^\infty \sum_{\gamma_1, \dots, \gamma_n \in \Gamma} \varphi(\gamma_1, \dots,\gamma_n) \prod_{i = 1}^n w(\gamma_i).
    \label{eq:MayerExplicit}
\end{equation}
Here $\Gamma$ denotes the set of all possible polymers $\gamma$, while $\varphi$ are \emph{Ursell functions}, defined as $\varphi(\gamma) = 1$ and 
\begin{equation}
    \varphi(\gamma_1,  \dots, \gamma_n) = \frac{1}{n!} \sum_{G \in \mathcal{G}_n^c} \prod_{(i,j) \in {E(G)}} \comm{\delta(\gamma_i,\gamma_j) - 1}
\end{equation}
for $n>1$. The sum is performed over all \emph{connected} graphs $G$ on $n$ vertices $\{1, \ldots , n \}$, with $\mathcal{G}^c_n$ the set of all such connected graphs, and $E(G)$ the edge set of $G$. Note that if the polymers $\gamma_1, \dots, \gamma_n$ have any mutually non-intersecting bipartition, then for any connected graph $G \in \mathcal{G}^c_n$, there will always be an edge $(i,j) \in E(G)$ such that $\delta(\gamma_i,\gamma_j) - 1 = 0$, ensuring that $\varphi(\gamma_1, \dots, \gamma_n) = 0$. Therefore, we can restrict the sum over polymers in Eq.~\eqref{eq:MayerExplicit} to \emph{connected} sets of polymers---i.e., connected clusters---that we compactly label by $\mathcal{C} = (\gamma_1,\ldots,\gamma_n)$, where $n$ can range from $1$ to $\infty$.
We can then write
\begin{equation}
    \log \mathcal{Z} = \sum_{\CC} \varphi(\CC) \prod_{i = 1}^{\text{length}(\CC)} w(\gamma_i),
\end{equation}
and read off, with the aid of Eq.~\eqref{eq:W_EM}, the coefficients
\begin{equation}
\label{eq:fc_endmatter}
  f(\CC) = \varphi(\CC) \prod_{i = 1}^{\text{length}(\CC)} e^{-\beta\abs{\gamma_i}}
\end{equation}
that define our parent Hamiltonian $H(\beta)$ [see Eq.~\eqref{eq:HbetafC}].

\textbf{\textit{Dual Ising Wavefunction.}}~From a dual description, our construction demonstrates that long-range ferromagnetic order and perimeter-law disorder parameter correlations can coexist in a gapped ground state. Explicitly, the no-go state $\ket{\psi_{00}(\beta)}$ is related by Wegner duality to the Ising wavefunction
\begin{equation}
    \ket{\varphi(\beta)} = \exp \qty{ \frac{\beta}{2} \sum_{\expval{pq}} Z_{p} Z_{q} } \ket{+} ,
\end{equation}
where $Z_{p}, X_{p}$ are a new set of Pauli operators defined on the dual lattice (i.e., at the centers of plaquettes $p,q$ of the original lattice), and $\ket{+}$ is the simultaneous $+1$ eigenstate of each $X_{p}$. This state is symmetric under the $\mathbb{Z}_{2}$ symmetry $\prod_{p} X_{p}$, and exhibits long-range ferromagnetic correlations $\expval{Z_{p} Z_{p'}} \sim \mathcal{O}(1)$ for $\beta > \beta_{c}$; i.e., the $\mathbb{Z}_{2}$ symmetry is spontaneously broken at large $\beta$. 

{Conventional quantum ferromagnets, such as the 2D transverse-field Ising model, exhibit \emph{area-law} scaling of the disorder parameter, $\expval*{\prod_{p \in \mathcal{R}} X_p} \sim e^{-\alpha \abs{\mathcal{R}}}$ for a large 2D region $\mathcal{R}$. This scaling is especially natural from a ``mean-field'' perspective, where the ferromagnetic ground state is approximated as a symmetry-broken product state. In contrast,} $\ket{\varphi(\beta)}$ exhibits \emph{perimeter-law} correlations of the disorder parameter, $\expval*{\prod_{p \in \mathcal{R}} X_{p}} \sim e^{-\mu \abs{\partial \mathcal{R}}}$. Indeed, the latter operators are directly related to the Wilson loops $X_{\mathcal{L}}$ of the deformed toric code via Wegner duality~\cite{Wegner_duality_1971}. 

In 1D gapped ground states, a rigorous theorem \cite{Levin_constraints_2020} forbids the coexistence of long-range order parameter correlations and perimeter-law (i.e., long-ranged) disorder parameter correlations, and a natural question is whether a similar result can hold in two dimensions. Our parent Hamiltonian $H^{(00)}(\beta)$ for $\ket{\psi_{00}(\beta)}$ is dual to a $\mathbb{Z}_{2}$-symmetric gapped parent Hamiltonian $H_{\text{dual}}(\beta)$ with an exact two-fold degeneracy at large $\beta$. This Hamiltonian demonstrates via counterexample that a result analogous to Ref.~\onlinecite{Levin_constraints_2020} cannot hold in 2D under the standard definition of local Hamiltonians, which allows for interactions decaying exponentially in their \emph{diameter} [see Eq.~\eqref{eq:locality_bound}]. 

Recently, Ref.~\onlinecite{McDonough_Zhang_lieb-robinson_2025} proved that weak $\mathbb{Z}_{2}$-symmetric perturbations of the 2D ferromagnetic fixed point yield ferromagnetic ground states exhibiting area-law disorder parameters, $\expval*{\prod_{p \in \mathcal{R}} X_{p}} \sim e^{-\alpha \abs{\mathcal{R}}}$, so long as these perturbations exhibit interactions which decay exponentially with their \emph{volume}. This result is consistent with our parent Hamiltonian $H_{\text{dual}}(\beta)$, which violates this exponential-in-volume locality constraint. Indeed, $H_{\text{dual}}(\beta)$ contains interaction terms of the form $\prod_{p \in \mathcal{R}} X_{p}$ which are only suppressed in their \emph{perimeter} $\abs{\partial \mathcal{R}}$, rather than their volume $\abs{\mathcal{R}}$. It is natural to conjecture that long-range ferromagnetic order and perimeter-law disorder parameters cannot coexist in 2D gapped ground states of exponential-in-volume local Hamiltonians.

\textbf{\textit{The Strongly Decohered Toric Code is Separable.}}~Let $\rho_{p} = [\prod_{e} \mathcal{E}_{e}]( \dyad{\TC} )$ be the mixed state obtained by applying a phase-flip channel $\mathcal{E}_{e}(\rho) = (1 - p) \rho + p Z_{e} \rho Z_{e}$ of strength $p$ to each qubit of a toric code state $\ket{\TC}$. Such a mixed state is well-known to undergo a mixed-state \emph{decodability} phase transition \cite{dennis_topological_2002,fan_diagnostics_2024} at a critical threshold $p_{c} \approx 0.109$, such that the initial state $\ket{\TC}$ can no longer be reliably recovered from $\rho_{p}$ for $p > p_{c}$. Reference~\onlinecite{ChenGrover_separability_2024} argued that this transition coincides with a \emph{separability} phase transition, whereupon $\rho_{p}$ can be expressed as a convex sum of short-range entangled states for $p > p_{c}$. Specifically, letting the initial state $\ket{\TC}$ be the `++' logical state, Ref.~\onlinecite{ChenGrover_separability_2024} demonstrated that $\rho_{p}$ admits a decomposition into a mixture of the states
\begin{equation}
    \ket{\psi_{\mathcal{L}}(K)} = X_{\mathcal{L}} \sum_{\qty{x_{e} = \pm 1}} \sqrt{ \mathcal{Z}^{\text{Ising}}_{K}(\qty{x_{e}}) } \ket{\qty{x_{e}}}_x ,
\end{equation}
where $K = - \frac{1}{2} \log[p / (1-p)]$, $\ket{\qty{x_{e}}}_{x}$ are the Pauli-$X$ basis states, and $\mathcal{Z}^{\text{Ising}}_{K}(\qty{x_{e}}) = \sum_{\qty{s_{v}}} e^{K \sum_{\expval{v v'}} x_{v v'} s_{v} s_{v'}}$ is the partition function of a bond-disordered Ising model with disorder realization $\qty{x_{e}}$. 

Similar to the deformed toric code state~\eqref{eqn:defnofdefTC}, the states $\ket{\psi_{\mathcal{L}}(K)}$ are 1-form symmetric under contractible loop operators $Z_{\tilde{\mathcal{L}}}$, exhibit condensed $m$ anyons for $p > p_{c}$, but maintain perimeter-law Wilson loops for all $p < \frac{1}{2}$. Therefore, although each $\ket{\psi_{\mathcal{L}}(K)}$ is topologically trivial for $p > p_{c}$, the results of Ref.~\onlinecite{Sahay2025FunkyIsing} call into question whether these states are indeed short-range entangled. Using identical cluster expansion techniques to those described in the main text, we answer this question in the affirmative for sufficiently large $p< \frac{1}{2}$~\cite{SOM}.

Explicitly, recalling the expression for the operator $\mathcal{Z}$ in the right-hand side of Eq.~\eqref{eq:Vdef1}, we see that $\ket{\psi_{\mathcal{L}}(K)}$ can be rewritten as
\begin{equation}
    \ket{\psi_{\mathcal{L}}(K)} \propto X_{\mathcal{L}} \sqrt{\mathcal{Z}} \ket{0} = X_{\mathcal{L}} e^{W / 2} \ket{0} ,
\end{equation}
where we have used the simple identity $\mathcal{Z} \ket{\qty{x_{e}}}_{x} = \mathcal{Z}^{\text{Ising}}_{K}(\qty{x_{e}}) \ket{\qty{x_{e}}}_{x}$ [see Eq.~\eqref{eq:Vdef1}]. Using an identical approach as for the deformed toric code state $\ket{\psi(\beta)}$, one can construct a local parent Hamiltonian $H_{\mathcal{L}}(K)$ for which $\ket{\psi_{\mathcal{L}}(K)}$ is the unique gapped ground state at sufficiently large $p$. Applying the results of Ref.~\cite{Yin_Lucas_low-density_2025,DeRoeck_Khemani_low-density_2025}, $\ket{\psi_{\mathcal{L}}(K)}$ can be prepared via finite-time unitary evolution with a local generator, and is therefore short-range entangled.

In the Supplemental Material \cite{SOM}, we note that there is a minor subtlety when the `++' initial state is replaced by a general toric code state $\ket{\TC}$. For certain choices of $\ket{\TC}$, the pure-state decomposition of $\rho_{p}$ is not completely trivial: instead, it contains an exponentially small (in trace norm) sector of gapless states. Such mixed states $\rho_p$ are therefore \emph{approximately} separable at large $p$; in particular, they are indistinguishable from a corresponding separable state in the thermodynamic limit.

\textbf{\textit{Classical High-Low Temperature Duality.}}~Our rewriting of the deformed toric code state $\ket{\psi(\beta)} = e^{W} \ket{0}$ as a deformation of the trivial paramagnet leads to a novel high/low temperature duality between the 2D classical Ising model and a 2D gauge theory with exponentially decaying interactions. This duality is valid at large $\beta$, where the Ising model lies in a ferromagnetic phase and the dual gauge theory lies within a confined phase.

Recall first from Eq.~\eqref{eqn:partition_function} that the norm of $\ket{\psi(\beta)}$ can be expressed as the partition function of a 2D Ising model, where the loop configurations $\mathcal{L}$ are understood as magnetic domain walls. On the other hand, by Eq.~\eqref{eq:V_Mayer}, the same norm can be written as
\begin{equation}
    \bra{\psi(\beta)} \ket{\psi(\beta)} = \bra{0} e^{2W} \ket{0} = \bra{0} e^{2 \sum_{\mathcal{C}} f(\mathcal{C}) X_{\mathcal{C}}} \ket{0} .
\end{equation}
By expanding $\ket{0} \propto \sum_{\qty{x_{e}}} \ket{\qty{x_{e}}}_x$ in the Pauli-$X$ basis, we find that the 2D Ising partition function is proportional to the following dual partition function with Ising spins $x_{e} = \pm 1$ on the edges of the square lattice:
\begin{equation}
    \bra{\psi(\beta)} \ket{\psi(\beta)} \propto \sum_{\qty{x_{e} = \pm 1}} e^{2 \sum_{\mathcal{C}} f(\mathcal{C}) x_{\mathcal{C}}}, \quad x_{\mathcal{C}} = \prod_{\gamma \in \mathcal{C}} \prod_{e \in \gamma} x_{e} .
\end{equation}
This classical partition function is invariant under the local $\mathbb{Z}_{2}$ symmetry transformation $x_{v v'} \mapsto s_{v} x_{v v'} s_{v'}$, where we've written $x_{v v'} \equiv x_{e}$ for $e = \expval{v v'}$, and $\qty{s_{v} = \pm 1}$. Consequently, the above partition function describes a pure $\mathbb{Z}_{2}$ gauge theory.  It is unclear whether this dual model can be extended to (or past) the critical point, since the cluster expansion generally fails to converge at sufficiently small $\beta$.

\makeatletter
\let\addcontentsline\old@addcontentsline
\makeatother
\include{SM.tex}

\end{document}

%% file: SM.tex
\makeatletter 
    
\renewcommand\onecolumngrid{
\do@columngrid{one}{\@ne}%
\def\set@footnotewidth{\onecolumngrid}
\def\footnoterule{\kern-6pt\hrule width 1.5in\kern6pt}%
}

\renewcommand\twocolumngrid{
        \def\footnoterule{
        \dimen@\skip\footins\divide\dimen@\thr@@
        \kern-\dimen@\hrule width.5in\kern\dimen@}
        \do@columngrid{mlt}{\tw@}
}%

\makeatother

\onecolumngrid
\newpage

\renewcommand{\thefigure}{S\arabic{figure}}
\renewcommand{\theequation}{S\arabic{equation}}
\renewcommand{\thetable}{S\Roman{table}}
\renewcommand{\thesection}{S\Roman{section}}

\setcounter{secnumdepth}{2}
\setcounter{equation}{0}
\setcounter{section}{0}
\setcounter{figure}{0}
\setcounter{table}{0}
\setcounter{page}{1}

\renewcommand*{\thefootnote}{\fnsymbol{footnote}}
\setcounter{footnote}{0}
\footnotetext{
\setlength{\parindent}{0pt}
\hypertarget{equalcontrib}{}
\href{mailto:nmanoj@caltech.edu}{nmanoj@caltech.edu} \\
}
\setcounter{footnote}{0}
\renewcommand*{\thefootnote}{\arabic{footnote}}
\begin{center}
\textbf{\large Supplemental Material For:
Gapped Parent Hamiltonians for the Strongly Deformed Toric Code}

\vspace{5mm}
Nandagopal Manoj,\textsuperscript{\hyperlink{equalcontrib}{\textasteriskcentered}} Zack~Weinstein, and Jason Alicea


\textit{\small Department of Physics and Institute for Quantum Information and  \\
\vspace{-0.5mm}
Matter, California Institute of Technology, Pasadena, California 91125, USA} \\
\vspace{-0.5mm}
{\small (Dated: \today)}
\end{center}

\thispagestyle{empty}

\noindent This supplemental material describes the technical statements in the Letter in more detail and, wherever applicable, provides rigorous proofs and bounds. 

\tableofcontents

\section{Setup and Notation}
We work on an $L_x \times L_y$ square lattice $\Lambda$ with periodic boundary conditions, where we assume $L_x\geq L_y$ without loss of generality. The qubits live on the edges, labeled by $e$. We write Pauli operators acting on these qubits as $X_e, Z_e$. We use $\mathcal{L}$ to denote closed (not necessarily connected) loop configurations in the direct lattice---that is, $\mathcal{L}$ is a subset of edges $e$ such that each vertex $v \in \Lambda$ is incident to an even number of edges in $\mathcal{L}$. For each $\mathcal{L}$, we define the operator
\begin{equation}
    X_{\mathcal{L}} \equiv \prod_{e\in {\mathcal{L}}} X_e,
\end{equation}
and denote by $\abs{\mathcal{L}}$ the number of edges in ${\mathcal{L}}$. 

The deformed toric code states are defined as
\begin{equation}
    \ket{\psi(\beta)} \propto e^{\frac{\beta}{2}\sum_e Z_e} \ket{\TC} ,
\end{equation}
where $\ket{\TC}$ is a simultaneous +1 eigenstate of the stabilizers $A_v \equiv \prod_{e \ni v} Z_e$ and $B_p \equiv \prod_{e \in p} B_p$. In the main text, we primarily focused on the state obtained by deforming the `++' logical state $\ket{\TC_{++}}$, given explicitly by
\begin{equation}
\label{eq:psi_++}
    \ket{\psi_{++}(\beta)} \propto e^{\frac{\beta}{2}\sum_e Z_e} \ket{\TC_{++}} \propto \sum_{\mathcal{L}} e^{-\beta \abs{\mathcal{L}}}X_{\mathcal{L}} \ket{0},
\end{equation}
where $\ket{0}$ is the simultaneous $+1$ eigenstate of each $Z_e$. In other words, $\ket{\psi_{++}(\beta)}$ contains a superposition of all closed loops of $Z_e = -1$ in the direct lattice, including both contractible and non-contractible loops.

Similarly, the ``no-go state'' is obtained by deforming the `00' logical state, and is given by
\begin{equation}
\label{eq:psi_00}
    \ket{\psi_{00}(\beta)} \propto e^{\frac{\beta}{2}\sum_e Z_e} \ket{\TC_{00}} \propto \sum_{{\mathcal{L}} = \partial \mathcal{R} } e^{-\beta \abs{\mathcal{L}}}X_{\mathcal{L}} \ket{0} .
\end{equation}
This state contains only contractible loops in the superposition, and consequently satisfies the assumptions of Theorem~1 in Ref.~\onlinecite{Sahay2025FunkyIsing}. We shall initially focus on the simpler state $\ket{\psi_{++}(\beta)}$, which we denote by $\ket{\psi(\beta)}$ to lighten the notation; we return to the no-go state $\ket{\psi_{00}(\beta)}$ in Sec.~\ref{sec:no_go}.

Finally, we define the gapped parent Hamiltonian for the trivial product state $\ket{0}$ by
\begin{equation}
    H_0 = \sum_e (1 - Z_e) ,
\end{equation}
which hosts a unique ground state with the spectral gap $\Delta_0 = 2$. 

\section{Polymer models and Mayer Cluster Expansion}
\label{sec:cluster_expansion}
In order to construct parent a Hamiltonian for the state $\ket{\psi(\beta)}$ at large $\beta$, it is insightful to view it as non-unitary deformation of the trivial state $\ket{0}$. From Eqs.~\eqref{eq:psi_++} we have $\ket{\psi(\beta)} = \hat{\mathcal{Z}}(\beta) \ket{0}$, where the operator $\hat{\mathcal{Z}}(\beta)$ is defined by
\begin{equation}
\label{eq:operator_partition_fns}
    \hat{\mathcal{Z}}(\beta) \equiv \sum_{\mathcal{L}} e^{-\beta \abs{\mathcal{L}}} X_{\mathcal{L}} .
\end{equation}
As described in the main text [see Eq.~\eqref{eq:Vdef1}], this operator can be expressed as the high-temperature expansion of a bond-disordered Ising model with coupling constant $K = \artanh(e^{-\beta})$, and is therefore a positive-definite operator. Its logarithm,
\begin{equation}
    \hat{W}(\beta) \equiv \log \hat{\mathcal{Z}}(\beta) ,
\end{equation}
is therefore a well-defined Hermitian operator. 

The locality of the parent Hamiltonian $H(\beta)$ we will construct depends sensitively on the locality properties of $\hat{W}(\beta)$. Heuristically, since $\hat{W}$ is the free energy of a bond-disordered Ising model, we expect it to be an extensive sum of terms which are exponentially localized at large $\beta$. We will formally demonstrate this property using the \emph{Mayer cluster expansion}, described extensively in Ref.~\onlinecite{Friedli_Velenik_2017}. Here we review the basics of this series expansion.

\subsection{Polymer model definition}
A \emph{polymer}, labeled by $\gamma$, is a collection of edges that form a \emph{connected} closed loop. Note that polymers can be self-intersecting, like a figure of eight. We define $\Gamma$ to be the set of all polymers on the $L_x \times L_y$ square lattice. A \emph{polymer model} is defined by associating a complex weight $w: \Gamma \to \mathbb{C}$ to each polymer and a real symmetric interaction $\delta: \Gamma \times \Gamma \to [-1, 1]$ to each pair of polymers. Given $(\Gamma, w,\delta)$, the polymer partition function is defined as 
\begin{equation}
    \mathcal{Z} = \sum_{\Gamma'\subseteq \Gamma} \qty[ \prod_{\gamma \in \Gamma'} w(\gamma) ]  \qty[ \prod_{\curly{\gamma,\gamma'} \subseteq \Gamma'} \delta(\gamma, \gamma') ] .
    \label{eqn:PolymerPartitionFunction}
\end{equation}
Throughout this work, we will largely focus on the weight function $w(\gamma) = e^{-\beta |\gamma|}$. To describe the operator-valued partition function \eqref{eq:operator_partition_fns}, we will later promote $w(\gamma)$ to the operator-valued weight function
\begin{equation}
\label{eq:app_operator_weights}
    \hat{w}(\gamma) \equiv e^{-\beta \abs{\gamma}} X_\gamma ,
\end{equation}
where $X_{\gamma} \equiv \prod_{e \in \gamma} X_e$. These operators mutually commute, so we can regard the operator-valued partition function $\hat{\mathcal{Z}}(\beta)$ obtained from these weights as an ordinary polymer partition function within each simultaneous Pauli-$X$ eigenstate.

Finally, it will be sufficient to work with a hard-core interaction where $\delta(\gamma, \gamma') \in \curly{0,1}$, defined presently\footnote{A slightly more intricate definition will be necessary to describe the no-go state; see Sec.~\ref{sec:no_go}.} as follows:
\begin{equation}
\label{eq:app_delta_fn}
    \delta(\gamma, \gamma') = \begin{cases}
        0, & \gamma, \gamma' \text{ intersect} \\
        1, & \text{otherwise}
    \end{cases} .
\end{equation}
Two polymers $\gamma, \gamma'$ are defined as intersecting if there exists edges $e \in \gamma$ and $e' \in \gamma'$ such that $e, e'$ share a common vertex in the lattice. Note that by this definition, two diagonally adjacent polymers which touch at a single corner are considered intersecting. Loop configurations such as `figure-eights' are therefore considered as a single polymer.

We say that two intersecting polymers $\gamma, \gamma'$ are \emph{incompatible}, as if the subset $\Gamma'$ in the partition sum contains two such polymers, the summand will be zero and will not contribute. We also say that $\gamma, \gamma'$ are \emph{compatible}, or \emph{non-interacting}, if $\delta(\gamma, \gamma') = 1$.

\subsection{Cluster Expansion and Ursell Functions}

A \textit{cluster} $\mathcal{C}$ is a finite non-empty sequence of polymers $(\gamma_1, \gamma_2, \dots, \gamma_n)$, with repetitions allowed. A cluster is called \emph{disconnected} if there exists a partition of the integers $\{1, \ldots , n \} = I_1 \sqcup I_2$ such that, for all $i \in I_1$ and $j \in I_2$, $\gamma_i$ is compatible with $\gamma_j$ (i.e., $\delta(\gamma_i, \gamma_j) = 1$); in other words, a disconnected cluster can be divided into two sub-clusters $\mathcal{C}_1$ and $\mathcal{C}_2$ such that each $\gamma_i \in \mathcal{C}_1$ is compatible with each $\gamma_j \in \mathcal{C}_2$. Finally, a cluster is \emph{connected} if it is not disconnected.

The Mayer cluster expansion \cite{Friedli_Velenik_2017,KoteckyPreiss1986,Penrose1967} provides an explicit series expansion for the polymer free energy $W = \log \mathcal{Z}$ as a sum over connected clusters:
\begin{theorem}[Mayer Cluster Expansion \cite{Friedli_Velenik_2017}] \label{theorem:Mayer}=

Given the polymer model $(\Gamma, w, \delta)$ and the corresponding polymer partition function $\mathcal{Z}$ [see Eq.~\eqref{eqn:PolymerPartitionFunction}], the polymer free energy $W \equiv \log \mathcal{Z}$ can be expressed as the following infinite series over clusters $\mathcal{C}$:
\begin{equation}
\label{eq:mayer_cluster}
    W \equiv \log \mathcal{Z} = \sum_{\mathcal{C}} f(\mathcal{C}), \quad f(\mathcal{C}) \equiv \varphi(\gamma_1 , \ldots , \gamma_n ) \prod_{i = 1}^n w(\gamma_i) .
\end{equation}
Here $\varphi(\gamma_1 , \ldots , \gamma_n)$ is the \emph{Ursell function}, defined for each cluster $\mathcal{C} = (\gamma_1 , \ldots , \gamma_n)$ as
\begin{equation}
\label{eq:ursell}
    \varphi(\gamma_1, \dots, \gamma_n) \equiv \frac{1}{n!} \sum_{G \in \mathcal{G}^c_n} \prod_{(i,j) \in E(G)} \zeta(\gamma_i,\gamma_j) ,
\end{equation}
where $\mathcal{G}^c_n$ is the set of all connected graphs $G$ with vertex set $V(G) = \{ 1, \ldots , n \}$, $E(G)$ is the edge set of $G$, and $\zeta(\gamma, \gamma') \equiv \delta(\gamma, \gamma') - 1$, so that $\zeta(\gamma, \gamma')$ vanishes if the polymers are compatible. For $n=1$ we define $\varphi(\gamma_1) = 1$. Note that the Ursell functions vanish on disconnected clusters, and so we can restrict the sum in Eq.~\eqref{eq:mayer_cluster} to connected clusters $\mathcal{C}$. For a derivation of Eq.~\eqref{eq:mayer_cluster} and details on its convergence, we refer the reader to Ref.~\onlinecite{Friedli_Velenik_2017}.
\end{theorem}

We can apply the cluster expansion to Eq.~(\ref{eq:operator_partition_fns}), using the weights \eqref{eq:app_operator_weights} and interactions \eqref{eq:app_delta_fn}. As the weights $\hat{w}(\gamma)$ are composed only of $X_e$ operators, we can ignore the operator nature of the weights by working in the $X_e$ eigenbasis. Then, we have 
\begin{equation}
\label{eq:cluster_expansion_op}
    \hat{W}(\beta) \equiv \log \hat{\mathcal{Z}}(\beta) = \sum_{n \geq 1} \sum_{\gamma_1, \dots, \gamma_n \in \Gamma} \varphi(\gamma_1, \dots, \gamma_n) \prod_{i=1}^n e^{-\beta \abs{\gamma_i}} X_{\gamma_i} \equiv \sum_{\mathcal{C}} f(\CC) X_{\CC},
\end{equation}
where $X_{\CC} \equiv\prod_{i=1}^n X_{\gamma_i}$ for any cluster $\mathcal{C} = (\gamma_1 , \ldots , \gamma_n)$. A crucial bound satisfied by these coefficients $f(\CC)$ is Eq.~\eqref{eqn:fCbound} in the main text (Lemma~\ref{lem:Cbound}). Since the proof of this bound is rather technical, we postpone its proof to Appendix~\ref{sec:proof_of_eq13}.

\section{Proof of Locality and Spectral Gap}
\label{app:locality_proof}
In this Appendix, we provide explicit technical details to prove that our parent Hamiltonian $H(\beta)$ is local and exhibits a nonzero spectral gap at large values of $\beta$. Our essential strategy is to use properties of the cluster expansion---specifically, the bound in Eq.~\eqref{eqn:fCbound} of the main text---to show that $H(\beta)$ satisfies the locality properties demanded by the assumptions of Refs.~\onlinecite{BravyiHastingsMichalakis2010,BravyiHastings2011}. The gap stability results of these works then immediately imply that $H(\beta)$ is gapped at sufficiently large $\beta$.

Let us first recall the derivation of our parent Hamiltonian $H(\beta)$. We begin by using the cluster expansion to express the deformed toric code state $\ket{\psi(\beta)}$ as an exponential deformation of the trivial state $\ket{0}$:
\begin{equation}
\label{eq:psi_deformed_app}
    \ket{\psi(\beta)} = \hat{\mathcal{Z}}(\beta) \ket{0} = e^{\hat{W}(\beta)} \ket{0} = \exp \qty{ \sum_{\mathcal{C}} f(\mathcal{C}) X_{\mathcal{C}} } \ket{0} ,
\end{equation}
c.f. Eqs.~\eqref{eq:operator_partition_fns} and \eqref{eq:cluster_expansion_op}. As described in the main text, this form naturally motivates the construction of a frustration-free parent Hamiltonian as follows: we first define the manifestly positive-semidefinite Hermitian operators $P_e(\beta)$ via
\begin{equation}
    \begin{split}
        P_e(\beta) &\equiv e^{- \sum_{\mathcal{C} \in \mathcal{S}_e} f(\mathcal{C}) X_{\mathcal{C}}} (1 - Z_e) e^{- \sum_{\mathcal{C} \in \mathcal{S}_e} f(\mathcal{C}) X_{\mathcal{C}}} \\
        &= e^{- 2\sum_{\mathcal{C} \in \mathcal{S}_e} f(\mathcal{C}) X_{\mathcal{C}}} - Z_e ,
    \end{split}
\end{equation}
where $\mathcal{S}_e$ denotes the set of connected clusters $\mathcal{C}$ for which $X_{\mathcal{C}}$ anticommutes with $Z_e$. Each $P_e(\beta)$ clearly annihilates $\ket{\psi(\beta)}$, and therefore the Hamiltonian
\begin{equation}
\label{eq:sm_parent_ham}
    H(\beta) \equiv \sum_e P_e(\beta) = \sum_e \qty( e^{- 2\sum_{\mathcal{C} \in \mathcal{S}_e} f(\mathcal{C}) X_{\mathcal{C}}} - Z_e )
\end{equation}
admits $\ket{\psi(\beta)}$ as a frustration-free ground state. In the same way that the Castelnovo-Chamon Hamiltonian $H_{\text{CC}}(\beta)$ [see Eq.~\eqref{eq:HCC} of the main text] can be viewed as a deformation of the fixed-point toric code Hamiltonian $H_{\TC}$, our parent Hamiltonian $H(\beta)$ can be viewed as a deformation of the trivial Hamiltonian $H_0 = \sum_e (1 - Z_e)$.

To prove a spectral gap, it will suffice to demonstrate that $H(\beta)$ is exponentially localized at sufficiently large $\beta$. More precisely, our goal is to write $H(\beta)$ in the form
\begin{equation}
\label{eq:hastsings_1}
    H(\beta) = H_0 + V ,
\end{equation}
where $H_0 = \sum_e (1 - Z_e)$ is the trivial Hamiltonian, while $V$ is a perturbation of the form
\begin{equation}
\label{eq:hastings_2}
    V = \sum_{r \geq 1} \sum_{A \in S(r)} V_{r,A} ,
\end{equation}
where $S(r)$ is the set of all $r \times r$ squares of lattice sites\footnote{For simplicity, we group the edges $e = (v, v + \hat{x})$ and $e = (v, v + \hat{y})$ together on the vertex $v$, so that each vertex is regarded as hosting two qubits.}, and $V_{r, A}$ is an operator supported on the $r \times r$ square $A$ with a spectral norm bounded as
\begin{equation}
\label{eq:hastings_3}
    \norm{V_{r, A}} \leq J e^{-\mu r} \quad \forall A \in S(r)
\end{equation}
for constants $J, \mu > 0$ to be determined. Note that this locality criterion is equivalent to that of Eq.~\eqref{eq:locality_bound} of the main text (see Appendix~\ref{app:locality_equivalence}). As we shall see, we will be able to choose $\mu = 1$ and $J = J(\beta) \propto e^{-4\beta}$, so that $J(\beta)$ becomes exponentially small at large $\beta$. The following theorem then guarantees that $\ket{\psi(\beta)}$ is the unique gapped ground state of $H(\beta)$ with a nonvanishing spectral gap:
\begin{theorem}[Stability of Gap~\cite{BravyiHastingsMichalakis2010}]\label{theorem:BHM}
    There exist constants $J_0, c_1 > 0$ depending only on $\mu$ and the spatial dimension $D$, such that for all $J\leq J_0$, the spectrum of $H_0 + V$ is contained (up to an overall energy shift) in the union of intervals $\cup_{k\geq 0} I_k$, where $k$ runs over the spectrum of $H_0$ and $I_k$ is the closed interval
\begin{equation}
    I_k = \comm{k(1-c_1 J) - \delta, k(1+c_1 J) + \delta} 
\end{equation}
for some $\delta$ bounded by $J$ times a quantity decaying faster than any power of the system size $L$.
\end{theorem}
We can ignore $\delta$ as we are ultimately interested in the thermodynamic limit. This theorem tells us that the spectral gap of $H(\beta)$ is lower-bounded by $2 - 2 c_1 J$, which is positive for sufficiently small $J$.

\subsection{Proving that \texorpdfstring{$H(\beta)$}{Hbeta} is local}
\label{subsec:locality_proof}
We now set out to prove that $H(\beta)$ is local at sufficiently large $\beta$, in the sense of Eqs.~\eqref{eq:hastsings_1}, \eqref{eq:hastings_2}, and \eqref{eq:hastings_3}. Theorem~\ref{theorem:BHM} will then immediately imply the existence of a spectral gap at sufficiently large values of $\beta$. The results of this section will crucially rely on the bound in Eq.~\eqref{eqn:fCbound} of the main text, which we repeat here for convenience: for $\beta > \beta^* = 2 + \log 3$, we have the inequality
\begin{equation}
\label{eq:app_fc_bound}
    \sum_{\CC \in \mathcal{S}_e} \abs{f(\CC)} \leq \Cclus e^{-4\beta} ,
\end{equation}
where $\Cclus \approx 5.115 \times 10^3$ is a constant. Since the proof of this bound is rather technical, we defer it to Appendix~\ref{sec:proof_of_eq13}.

We begin by Taylor expanding the exponential in Eq.~\eqref{eq:sm_parent_ham} [justified by the bound in Eq.~\eqref{eq:app_fc_bound}] to write
\begin{equation}
    \begin{split}
        H(\beta) &= H_0 + \sum_e \sum_{n \geq 1} \frac{(-2)^n}{n!} \qty[ \sum_{\mathcal{C} \in \mathcal{S}_e} f(\mathcal{C}) X_{\mathcal{C}} ]^n  \\
        &= H_0 + \sum_e \sum_{n \geq 1} \frac{(-2)^n}{n!} \sum_{\mathcal{C}_1 \ldots \mathcal{C}_n \in \mathcal{S}_e} \prod_{j = 1}^n \qty[ f(\mathcal{C}_j) X_{\mathcal{C}_j} ] .
    \end{split}
\end{equation}
Our goal is to express the second term above in the form of Eq.~\eqref{eq:hastings_2}. For each vertex $v$ and odd values of $r$, let $A_v(r) \in S(r)$ denote the set of vertices contained within the $r \times r$ square centered on $v$. We then define the operator $V_{r,v}$ supported on $A_v(r)$ via
\begin{equation}
\label{eq:V_rv_app}
    V_{r, v} \equiv \sum_{e = (v, v + \hat{x}), (v, v + \hat{y})} \sum_{n \geq 1} \frac{(-2)^n}{n!} \sum_{\substack{\mathcal{C}_1 \ldots \mathcal{C}_n \in \mathcal{S}_e \\ \max_j \norm{\CC_j} = r+1}} \prod_{j = 1}^n [f(\mathcal{C}_j) X_{\mathcal{C}_j}] ,
\end{equation}
where $\norm{\mathcal{C}} \equiv \sum_{\gamma \in \mathcal{C}} \abs{\gamma}$ denotes the total length of all polymers (including repetitions) contained in the cluster $\mathcal{C}$. In the above, we have restricted the sum over clusters $\mathcal{C}_1 \ldots \mathcal{C}_n$ such that the maximum value of any $\norm{ \mathcal{C}_j }$ is precisely $r + 1$. It is straightforward to show\footnote{Indeed, the ``worst-case scenario'' arises when a cluster $\mathcal{C}_j$ consists of a single closed-loop polymer of width 1 and length $\lfloor r/2 \rfloor$, whose perimeter is precisely $r+1$.} that with this restriction, $V_{r, v}$ is supported on the $r \times r$ square $A_v(r)$. We then have the decomposition\footnote{For $r \geq L_y = \min(L_x, L_y)$, one may worry that many different squares $A_v(r) \in S(r)$ (namely, shifts of $v$ in the $y$ direction) actually describe the same set of lattice points. This is not an issue, since the bound we will prove demonstrates that the spectral norms of resulting operators $V_{r,v}$ are exponentially suppressed in $r \geq L_y$. By Weyl's inequality, such terms cannot shift the eigenvalues of $H(\beta)$ by more than an $e^{-\mathcal{O}(L_y)}$ amount, and we can therefore neglect them in proving the stability of the gap for sufficiently large system sizes.}
\begin{equation}
    H(\beta) = H_0 + \sum_{v} \sum_{r \geq 3, \text{odd}} V_{r, v} ,
\end{equation}
which matches the desired form in Eq.~\eqref{eq:hastings_2}.

Next, we want to bound the spectral norm of each $V_{r,v}$ as in Eq.~\eqref{eq:hastings_3}. Since each $\prod_j X_{\mathcal{C}_j}$ has a spectral norm of 1, the triangle inequality immediately gives
\begin{subequations}
\label{eq:vbound_app}
\begin{align}
    \norm{V_{r,v}} &\leq  \sum_{e = \term{v,v+\hat{x}},\term{v, v+\hat{y}}} \sum_{n \geq 1}  \frac{2^n}{n!} \sum_{\substack{\CC_1 \dots \CC_n \in \mathcal{S}_e  \\ \max_j \norm{\CC_j} = (r+1)}} \prod_{j=1}^n { \abs{f(\CC_j)} } \\
    &\leq  \sum_{e = \term{v,v+\hat{x}},\term{v, v+\hat{y}}} \sum_{n \geq 1}  \frac{2^n}{n!} \,  \sum_{k=1}^n \, \sum_{\substack{\CC_1 \dots \CC_n \in \mathcal{S}_e \\ \norm{\CC_k} = (r+1)}} \prod_{j=1}^n {\abs{f(\CC_j)} } \label{eq:vbound_2} \\
    &=  \sum_{e = \term{v,v+\hat{x}},\term{v, v+\hat{y}}} \sum_{n \geq 1}  \frac{2^n}{(n-1)!} \term{ \sum_{\substack{\CC \in \mathcal{S}_e \\ \norm{\CC} = r+1}} \abs{f(\CC)}} \term{ \sum_{\CC \in \mathcal{S}_e}  \abs{f(\CC)}}^{n-1} . \label{eq:vbound_3}
\end{align}
\end{subequations}
To obtain Eq.~\eqref{eq:vbound_2}, we replace the maximization over $\CC_j$ with a sum over all possible $\CC_k$ for which $\norm{\CC_k} = r + 1$; Eq.~\eqref{eq:vbound_3} then follows immediately by noting that the summand for each $k$ is identical.

The latter sum over clusters is immediately bounded using Eq.~\eqref{eq:app_fc_bound}. To bound the former sum over clusters, we note from Eq.~\eqref{eq:fc_endmatter} that $\abs{f(\CC)} = e^{-\beta \norm{\CC} / 2} \abs{f_{\beta / 2}(\mathcal{C})}$, where $f_{\beta/2}(\mathcal{C})$ is the same function $f(\mathcal{C})$ evaluated at the deformation strength $\beta / 2$ instead of $\beta$. If we assume for simplicity that $\beta \geq 2 \beta^*$, we can once again use Eq.~\eqref{eq:app_fc_bound} to bound the former sum over clusters as follows:
\begin{subequations}
\label{eq:vbound_app_2}
    \begin{align}
        \norm{V_{r, v}} &\leq \sum_{e = (v, v + \hat{x}), (v, v + \hat{y})} \sum_{n \geq 1} \frac{2^n}{(n-1)!} \qty( e^{-\beta (r + 1) / 2} \sum_{\substack{C \in \mathcal{S}_e \\ \norm{\CC} = r + 1}} \abs{f_{\beta / 2}(\mathcal{C})} ) \qty( \sum_{C \in \mathcal{S}_e} \abs{f(\mathcal{C})} )^{n-1} \\
        &\leq 2e^{- \beta (r + 1) / 2} \qty( \Cclus e^{-2 \beta} ) \sum_{n \geq 1} \frac{2^n}{(n-1)!} \qty( \Cclus e^{-4\beta} )^{n-1} \\
        &= 4 e^{-\beta r / 2} \Cclus e^{-5\beta / 2} \exp[ 2 \Cclus e^{-4\beta} ] \\
        &= \qty( 4 \Cclus e^{-4\beta} \exp \qty[ 2 \Cclus e^{-4\beta} ] ) e^{-\beta(r - 3) / 2} \\
        &\leq \qty( 4 \Cclus e^{-4\beta} \exp \qty[ 2 \Cclus e^{-8\beta^*} ] ) e^{-(r-3)} ,
    \end{align}
\end{subequations}
where we've used $\beta / 2 \geq \beta^* \geq 1$ inside the $r$-dependent exponential, recalling that $r \geq 3$. Altogether, we have proven the bound
\begin{equation}
\label{eq:vbound_app_final}
    \norm{V_{r,v}} \leq \tilde{J} e^{-4\beta} e^{-r}, \quad \tilde{J} = 4\Cclus \exp \qty[ 3 + 2\Cclus e^{-8\beta^*} ]  \approx 4.11 \times 10^5 .
\end{equation}
We have therefore proven that our parent Hamiltonian $H(\beta) = H_0 + V$ satisfies the locality bound in Eq.~\eqref{eq:hastings_2}, with $\mu = 1$. Since $J(\beta) \equiv \tilde{J} e^{-4\beta}$ can be made arbitrarily small at large $\beta$ (eg. for $\beta = 5, J(\beta) \approx 10^{-3}$), it follows from Theorem~\ref{theorem:BHM} that $H(\beta)$ is gapped and nondegenerate for sufficiently large $\beta$. Moreover, since $H(\beta)$ is continuously connected to the trivial Hamiltonian $H_0$ by a gap-preserving deformation, we also find that $\ket{\psi(\beta)}$ lies in the trivial phase at large $\beta$.

\section{Proof of Cluster Expansion Bound}
\label{sec:proof_of_eq13}
Having proven the locality of our parent Hamiltonian $H(\beta)$ utilizing the bound in Eq.~\eqref{eqn:fCbound} of the main text, we now return and demonstrate this bound. We will prove this result in several steps, using techniques which are standard to the cluster expansion \cite{Friedli_Velenik_2017,KoteckyPreiss1986,Penrose1967,Kruskal1956,KleinbergTardos2006,FernandezProcacci2007}.

\begin{lemma}[Tree-graph identity for the Ursell function \cite{Penrose1967}]
    For any $n \geq 1$ and polymers $\gamma_1,\dots,\gamma_n$ (repetitions allowed), the Ursell functions \eqref{eq:ursell} can be expressed as a sum over tree graphs as follows:
    \begin{equation}
    \label{eq:ursell_tree_graph}
        \varphi(\gamma_1,\dots,\gamma_n) = \frac{1}{n!} \sum_{T \in \mathcal{T}_n} \term{\prod_{(i,j) \in E(T)}\zeta(\gamma_i,\gamma_j) } \term{ \prod_{(i,j) \in \mathcal{E}_{\text{extra}}(T)} \delta(\gamma_i,\gamma_j) } ,
    \end{equation}
    where $\mathcal{T}_n$ is the set of spanning trees over $n$ labeled vertices, and $\mathcal{E}_{\mathrm{extra}}$ is defined below.
\end{lemma}

\begin{proof}
    This result was first derived in Ref.~\onlinecite{Penrose1967} to show convergence of the cluster expansion. For completeness, we include a proof here.

First, we describe Kruskal's algorithm~\cite{Kruskal1956,KleinbergTardos2006} , which constructs a unique spanning tree $T \in \mathcal{T}_n$ given any $G \in \mathcal{G}_n^c$. We start by defining an ordering over all possible edges in $G$ (we will call this set $E_n$). We choose the ordering 
\begin{equation}
    (1,2), \ (1,3), \ \dots,(1,n), \ (2,3), \ (2,4), \ \dots, \ (2,n), \ (3,4), \ \dots, \  (n-1,n).
\end{equation}
Kruskal's algorithm then defines a map $K: \mathcal{G}_n^c \to \mathcal{T}_n$ which goes through each link in $G \in \mathcal{G}_n^c$ in the above order, and adds the link to the resulting tree if it does not create a cycle, skipping the link otherwise. While $K$ assigns a unique $T \in \mathcal{T}_n$ to each $G \in \mathcal{G}^c_n$, note that $K$ is not injective; several different connected graphs $G$ are mapped to each tree $T$.

Using Kruskal's algorithm, we can reorder the sum as
\begin{subequations}
\begin{align}
    \varphi(\gamma_1, \dots, \gamma_n) &= \frac{1}{n!} \sum_{T \in \mathcal{T}_n} \sum_{G \in \mathcal{G}_n^c : K(G) = T} \prod_{(i,j) \in E(G)} \zeta(\gamma_i,\gamma_j) \\
    &= \frac{1}{n!} \sum_{T \in \mathcal{T}_n} \term{\prod_{(i,j) \in E(T)}\zeta(\gamma_i,\gamma_j) } \sum_{G \in \mathcal{G}_n^c : K(G) = T} \prod_{(i,j) \in E(G) \setminus E(T)} \zeta(\gamma_i,\gamma_j) .
\end{align}
\end{subequations}
Next, consider the set $E(G) \setminus E(T)$. By the construction of Kruskal's algorithm, an edge is present in $G$ but excluded from $T$ if and only if its addition would create a cycle with the edges already selected. Consequently, there exists a set of edges $\mathcal{E}_{\text{extra}}(T) \subseteq E_n$ that would be skipped over in Kruskal's algorithm for any input that resulted in $T$. Therefore, instead of summing over all graphs such that $K(G) = T$, we can equivalently sum over all subsets of $\mathcal{E}_{\text{extra}}(T)$:
\begin{subequations}
\begin{align}
    \sum_{G \in \mathcal{G}_n^c : K(G) = T} \prod_{(i,j) \in E(G) \setminus E(T)} \zeta(\gamma_i,\gamma_j) &= \sum_{\mathcal{E} \subseteq \mathcal{E}_{\text{extra}}(T)} \prod_{(i,j) \in \mathcal{E}} \zeta(\gamma_i, \gamma_j) \\
    &= \prod_{(i,j) \in \mathcal{E}_{\text{extra}}(T)} [1 + \zeta(\gamma_i,\gamma_j)] \\
    &= \prod_{(i, j) \in \mathcal{E}_{\text{extra}(T)}} \delta(\gamma_i, \gamma_j) ,
\end{align}
\end{subequations}
from which Eq.~\eqref{eq:ursell_tree_graph} follows.
\end{proof}

A useful feature of this rewriting is that every spanning tree $T \in \mathcal{T}_n$ has precisely $n-1$ edges; consequently, every term in Eq.~\eqref{eq:ursell_tree_graph} contributes with the same sign [unlike the original expression in Eq.~\eqref{eq:ursell}]. We therefore have the immediate corollary:
\begin{equation}
\label{eq:ursell_tree_graph_abs}
    \abs{\varphi(\gamma_1, \ldots , \gamma_n)} = \frac{1}{n!} \sum_{T \in \mathcal{T}_n} \qty( \prod_{(i, j) \in E(T)} \abs{\zeta(\gamma_i, \gamma_j)} ) \qty( \prod_{(i, j) \in \mathcal{E}_{\text{extra}}(T)} \delta(\gamma_i, \gamma_j) ) .
\end{equation}

\begin{lemma}[Kotecky-Preiss inequality~\cite{KoteckyPreiss1986}]
\label{lemma:KP} 
When $\beta \geq \beta^\ast = 2 +\log 3$, the weight $w(\gamma) = e^{-\beta \abs{\gamma}}$ satisfies
\begin{equation}
    \sum_{\gamma' \nsim \gamma} w(\gamma') e^{|\gamma'|} \leq \abs{\gamma},
\end{equation}
where $\gamma' \nsim \gamma$ indicates the sum is performed over all polymers $\gamma'$ which are incompatible with the fixed polymer $\gamma$ (i.e., such that $\delta(\gamma, \gamma') = 0$).
\end{lemma}

\begin{proof}
The number of polymers on the square lattice of length $s$ containing a given vertex $v$ is upper bounded by $3^s$. 
Therefore, the number of incompatible $\gamma'$ of fixed length $\abs{\gamma'} = s$ is upper-bounded by $\abs{\gamma} 3^{s}$, since $\gamma'$ has no more than $\abs{\gamma}$ options where it can intersect $\gamma$. Then,
\begin{subequations}
    \begin{align}
    \sum_{\gamma' \nsim \gamma} w(\gamma') e^{|\gamma'|} &\leq \sum_{s\geq 4, \mathrm{even}} \abs{\gamma} 3^s e^{-\beta s} e^{s}  \\
    &= \frac{(3e^{1-\beta})^4}{1 - (3e^{1-\beta})^2} \abs{\gamma} \\
    &\leq \frac{e^{-4}}{1 - e^{-2}} e^{-4(\beta - \beta^\ast)} \abs{\gamma} \\
    &\leq \abs{\gamma}
\end{align}
\end{subequations}
where the first sum converges as long as $ \beta > \beta^\ast -1 = 1 + \log 3$. The parameter $\beta^*$ is chosen for convenience, so that the numerical prefactor is smaller than unity.
\end{proof}

\begin{lemma}[Rooted tree recursion]\label{lem:rooted}
Define, for any polymer $\gamma$, the {absolute rooted cluster weight}
\begin{equation}
    g(\gamma) = \sum_{n \geq 1} n \sum_{\gamma_2,\dots,\gamma_n} \abs{\varphi(\gamma,\gamma_2, \dots,\gamma_n)} \prod_{i = 1}^n w(\gamma_i),\qquad \gamma_1 \equiv \gamma.
\end{equation}
Then $g(\gamma)$ satisfies the recursive inequality 
\begin{equation}
\label{eq:rooted_tree_recursion}
    g(\gamma) \leq w(\gamma) e^{\sum_{\gamma' \nsim \gamma} g(\gamma')}.
\end{equation}
\end{lemma}
\begin{proof}
    We refer the reader to~\cite{KoteckyPreiss1986,FernandezProcacci2007,Ueltschi2003} for detailed discussions (including regarding convergence), leaving only the essential derivation here. Using Eq.~\eqref{eq:ursell_tree_graph_abs}, we have 
    \begin{equation}
        g(\gamma_1) = \sum_{n \geq 1} \frac{1}{(n-1)!} \sum_{\gamma_2,\dots,\gamma_n} \sum_{T \in \mathcal{T}_n} \term{\prod_{(i,j) \in E(T)}\abs{\zeta(\gamma_i,\gamma_j)} } \term{ \prod_{(i,j) \in \mathcal{E}_{\text{extra}}(T)} \delta(\gamma_i,\gamma_j) }  \prod_{i = 1}^n w(\gamma_i) .
        \label{eqn:gtree1}
    \end{equation}

    To derive Eq.~\eqref{eq:rooted_tree_recursion}, we regard each spanning tree $T \in \mathcal{T}_n$ as a \emph{rooted} tree, with vertex 1 as the root. If vertex 1 is connected to $k$ other vertices $c_1, \ldots , c_k$ by edges in $E(T)$, then removing vertex 1 splits the rooted tree into $k$ disjoint rooted subtrees $T_1, \ldots , T_k$ with $T_r$ rooted at $c_r$. Letting $V_r$ denote the vertex set of $T_r$, we note that $V_1, \ldots , V_k$ forms a partition of $\{ 2, \ldots , n \}$.

    We now simplify $g(\gamma_1)$ as follows. First, we can write the product over edges $(i, j) \in E(T)$ as
    \begin{equation}
        \prod_{(i, j) \in E(T)} \abs{\zeta(\gamma_i, \gamma_j)} = \prod_{r = 1}^k \qty[ \abs{\zeta(\gamma_{1}, \gamma_{c_r})} \prod_{(i, j) \in E(T_r)} \abs{\zeta(\gamma_i, \gamma_j)} ] .
    \end{equation}
    Similarly, the product over weights trivially factorizes:
    \begin{equation}
        \prod_{i = 1}^n w(\gamma_i) = w(\gamma_1) \prod_{r = 1}^k \prod_{i \in V_r} w(\gamma_i) .
    \end{equation}
    To handle the product over $\mathcal{E}_{\text{extra}}(T)$, we will use the inequality
    \begin{equation}
        \prod_{(i, j) \in \mathcal{E}_{\text{extra}}(T)} \delta(\gamma_i, \gamma_j) \leq \prod_{r = 1}^k \prod_{(i, j) \in \mathcal{E}_{\text{extra}}(T_r)} \delta(\gamma_i, \gamma_j) .
    \end{equation}
    This inequality follows from noting that $\mathcal{E}_{\text{extra}}(T_r) \subseteq \mathcal{E}_{\text{extra}}(T)$ for each subtree $T_r$, so the former product contains extra delta functions which can cause the left-hand side to vanish. Finally, we organize the sum over trees $T \in \mathcal{T}_n$ as a sum over all possible subtrees $T_1 \ldots T_k$: for some function $F(T)$ which only depends on the graph structure of $T$ (i.e., $F(T) = F(T')$ for isomorphic trees $T, T'$), we have
    \begin{equation}
        \sum_{T \in \mathcal{T}_n} F(T) = \sum_{k \geq 0} \frac{1}{k!} \sum_{n_1 + \ldots + n_k = n-1} \frac{(n-1)!}{n_1! \ldots n_k!} n_1 \ldots n_k \sum_{T_1 \in \mathcal{T}_{n_1}} \ldots \sum_{T_k \in \mathcal{T}_{n_k}} F(T = (T_1 \ldots T_k)) ,
    \end{equation}
    where $k$ sums over the number of possible subtrees, and $n_1 \ldots n_k$ sums over the number of vertices in each subtree. The multinomial coefficient $(n-1)! / n_1! \ldots n_k!$ counts the number of ways of sorting vertices $\{2, \ldots , n\}$ into the $k$ subtrees, and the factor $1/k!$ accounts for the fact that the ordering of the $k$ subtrees is immaterial. The product $n_1 \ldots n_k$ accounts for the $n_r$ choices of root vertex in each subtree. Finally, the notation $T = (T_1 \ldots T_k)$ indicates that the rooted tree $T \in \mathcal{T}_n$ is constructed by joining the roots of the subtrees $T_r$ as children to a new root. 
    
    Putting all of these facts together, we arrive at the inequality
    \begin{equation}
        \begin{split}
            g(\gamma_1) &\leq w(\gamma_1) \sum_{k \geq 0} \frac{1}{k!} \prod_{r = 1}^k \qty[ \sum_{n_r \geq 1} \frac{1}{(n_r - 1)!} \sum_{\gamma_1' \ldots \gamma_{n_r}'} \sum_{T_r \in \mathcal{T}_{n_r}} \abs{\zeta(\gamma_1, \gamma_1')} \qty( \prod_{(i, j) \in E(T_r)} \abs{\zeta(\gamma_i', \gamma_j')} ) \qty( \prod_{(i, j) \in \mathcal{E}_{\text{extra}}(T_r)} \delta(\gamma_i', \gamma_j')  ) \prod_{i = 1}^{n_r} w(\gamma_i') ] \\
            &= w(\gamma_1) \sum_{k \geq 0} \frac{1}{k!} \prod_{r = 1}^k \qty[ \sum_{\gamma_1'} \abs{\zeta(\gamma_1, \gamma_1')} g(\gamma_1') ] \\
            &= w(\gamma_1) \exp \qty{ \sum_{\gamma_1'} \abs{\zeta(\gamma_1, \gamma_1')} g(\gamma_1') }.
        \end{split}
    \end{equation}
    Finally noting that $\abs{\zeta(\gamma_1, \gamma_1')}$ is unity when $\gamma_1 \nsim \gamma_1'$ and zero otherwise, we arrive at Eq.~\eqref{eq:rooted_tree_recursion}.

\end{proof}

\begin{lemma}[Absolute rooted cluster bound~\cite{FernandezProcacci2007}]
\label{cor:absrootedbound}
    The absolute rooted cluster weight is bounded above as 
    \begin{equation}
    \label{eq:absolute_rooted_bound}
        g(\gamma) \leq w(\gamma) e^{|\gamma|} .
    \end{equation}
\end{lemma}

\begin{proof}
    We begin by defining, for any function $h : \Gamma \to \mathbb{R}$ on the polymers, a new function $\mathcal{J}[h] : \Gamma \to \mathbb{R}$ as follows:
    \begin{equation}
        \mathcal{J}[h](\gamma) \equiv w(\gamma) e^{\sum_{\gamma' \nsim \gamma} h(\gamma)} .
    \end{equation}
    We also define a partial ordering on the space of such functions: for any two real-valued functions $h, h'$ on the polymers, we write $h \leq h'$ if $h(\gamma) \leq h'(\gamma)$ for all $\gamma \in \Gamma$.
    
    Note that the functional map $\mathcal{J}$ has the following two properties:
    \begin{enumerate}
        \item $\mathcal{J}$ is \emph{monotone}: if $h \leq h'$, then $\mathcal{J}[h] \leq \mathcal{J}[h']$.
        \item the function $b(\gamma) \equiv w(\gamma) e^{|\gamma|}$ is an \emph{upper-barrier}: for any $h \leq b$, we have
        \begin{equation}
            \mathcal{J}[h](\gamma) \leq \mathcal{J}[b](\gamma) = w(\gamma) \exp \qty{ \sum_{\gamma' \nsim \gamma} w(\gamma) e^{|\gamma|} } \leq w(\gamma) e^{|\gamma|} = b(\gamma) ,
        \end{equation}
        where we have used the result of Lemma~\ref{lemma:KP}. In other words, whenever $h \leq b$, we also have $\mathcal{J}[h] \leq b$.
    \end{enumerate}
To prove Eq.~\eqref{eq:absolute_rooted_bound}, let us recursively define the sequence of functions $h_n : \Gamma \to \mathbb{R}$ by $h_0(\gamma) \equiv 0$ and $h_d \equiv \mathcal{J}[h_{d-1}]$ for $d \geq 1$; from the above two properties of $\mathcal{J}$, we immediately have $h_0 \leq h_1 \leq h_2 \leq \ldots \leq b$. We also define the functions
\begin{equation}
    g_d(\gamma_1) \equiv \sum_{n \geq 1} \frac{1}{(n-1)!} \sum_{\gamma_2,\dots,\gamma_n} \sum_{\substack{T \in \mathcal{T}_n, \\ \text{depth}(T) \leq d}} \term{\prod_{(i,j) \in E(T)}\abs{\zeta(\gamma_i,\gamma_j)} } \term{ \prod_{(i,j) \in \mathcal{E}_{\text{extra}}(T)} \delta(\gamma_i,\gamma_j) }  \prod_{i = 1}^n w(\gamma_i) ,
\end{equation}
which are identical to the absolute rooted cluster weight $g(\gamma)$ [see Eq.~\eqref{eqn:gtree1}], save for the restriction to trees $T$ of depth less than or equal to $d$. So long as Eq.~\eqref{eqn:gtree1} converges, we have $\lim_{d \to \infty} g_d(\gamma) = g(\gamma)$. 

We now make the following observations:
\begin{enumerate}
    \item For $d = 1$, we have $h_1(\gamma) = g_1(\gamma) = w(\gamma)$.
    \item For $d = 2$, we have $h_2(\gamma) = w(\gamma) e^{\sum_{\gamma' \nsim \gamma} w(\gamma')}$. On the other hand, by dropping the $\mathcal{E}_{\text{extra}}$ product in the definition of $g_2$, we obtain the inequality
    \begin{equation}
        g_2(\gamma_1) \leq \sum_{n \geq 1} \frac{1}{(n-1)!} \sum_{\gamma_2, \ldots , \gamma_n \nsim \gamma_1} \prod_{j = 1}^n w(\gamma_j) = w(\gamma_1) e^{\sum_{\gamma' \nsim \gamma_1} w(\gamma')} = h_2(\gamma_1) ,
    \end{equation}
    i.e., $g_2 \leq h_2$.
    \item More generally, suppose $g_{n - 1} \leq h_{n-1}$. Then, by following an identical strategy to the proof of Lemma~\ref{lem:rooted}, we obtain the inequality
    \begin{equation}
        g_d(\gamma) \leq w(\gamma) e^{\sum_{\gamma' \nsim \gamma} g_{d-1}(\gamma')} \leq w(\gamma) e^{\sum_{\gamma' \nsim \gamma} h_{d-1}(\gamma')} = h_d(\gamma) ,
    \end{equation}
    where the first inequality follows from noting that the subtrees $(T_1 , \ldots , T_k)$ of a depth-$d$ tree $T$ each have depth $d-1$. By induction, we therefore learn that $g_d \leq h_d$ for all $d$. In the limit $d \to \infty$, we obtain the result
    \begin{equation}
        g(\gamma) \leq \lim_{d \to \infty} h_d(\gamma) \leq b(\gamma) = w(\gamma) e^{| \gamma | } ,
    \end{equation}
    as desired.
\end{enumerate}
\end{proof}

\begin{lemma}\label{lem:Cbound}
    Let $e$ be an arbitrary edge on the lattice $\Lambda$. When $\beta \geq \beta^\ast = 2 + \log 3$, there exists $\Cclus >0 $ such that 
    \begin{equation}
        \sum_{\CC \in \mathcal{S}_e} \abs{f(\CC)} \leq \Cclus e^{-4\beta} ,
    \end{equation}
    where $C_{\text{clus}}$ is a constant defined below.
\end{lemma}
\begin{proof} 
Using Lemma~\ref{cor:absrootedbound}, 
\begin{subequations}
\begin{align}
    \sum_{\CC \in \mathcal{S}_e} \abs{f(\CC)} &\leq \sum_{n\geq 1} \sum_{\substack{\gamma_1,\dots,\gamma_n\\ \exists \gamma_k \ni e}} \abs{\varphi(\gamma_1,\dots,\gamma_n)} \prod_{j = 1}^n w(\gamma_j) \\
    &\leq  \sum_{n\geq 1} n \sum_{\gamma_1 \ni e}\sum_{\gamma_2,\dots,\gamma_n} \abs{\varphi(\gamma_1,\dots,\gamma_n)} \prod_{j = 1}^n w(\gamma_j) \\
    &= \sum_{\gamma_1 \ni e} g(\gamma_1)\\
    &\leq  \sum_{\gamma \ni e} w(\gamma) e^{|\gamma|} \\
    &\leq \sum_{s \geq 4, \mathrm{even}} 3^s e^{-\beta s} e^s \\
    &= \frac{(3e^{1-\beta})^4}{1 - (3e^{1-\beta})^2} \\
    &= \Cclus e^{-4\beta}
\end{align}
\end{subequations}
where we've noted that the number of polymers of length $s$ containing the edge $e$ is upper-bounded by $3^s$. In the last line we identify the constant $\Cclus \equiv { (3e)^4 } / (1-e^{-2}) \approx 5.115 \times 10^3$.
\end{proof}

\section{The 00 sector: deriving the parent Hamiltonian with aspect-ratio dependent locality bounds}
\label{sec:no_go}
Having proven that the deformed toric code state $\ket{\psi(\beta)} \equiv \ket{\psi_{++}(\beta)}$ admits a gapped parent Hamiltonian, we now return to the ``no-go'' state $\ket{\psi_{00}(\beta)}$; see Eq.~\eqref{eq:psi_00}. As noted in the main text, the slightly simpler state $\ket{\psi_{++}(\beta)}$ does not strictly satisfy the assumptions of Ref.~\onlinecite{Sahay2025FunkyIsing}'s Theorem 1, which requires an exact 1-form symmetry about both contractible and non-contractible cycles of the torus. In contrast, $\ket{\psi_{00}(\beta)}$ does exhibit an exact 1-form symmetry about all cycles, and therefore satisfies Ref.~\onlinecite{Sahay2025FunkyIsing}'s assumptions.

We begin by comparing the no-go state $\ket{\psi_{00}(\beta)}$ to the state $\ket{\psi_{++}(\beta)}$ we have studied thus far. From Eqs.~\eqref{eq:psi_++} and \eqref{eq:psi_00}, it is clear that these two states only differ by the inclusion of non-contractible loops in the latter state $\ket{\psi_{++}(\beta)}$. Deep in the strongly deformed phase, large loops are rare, and the relative amplitude of these non-contractible loops is suppressed exponentially in linear system size $L$. Indeed, a simple Peierls argument \cite{Friedli_Velenik_2017} allows for the overlap between the two states to be lower-bounded for large $\beta$ as
\begin{align}
    \frac{\abs{\braket{\psi_{00}(\beta)}{\psi_{++}(\beta)}}^2}{\braket{\psi_{00}(\beta)}{\psi_{00}(\beta)}\braket{\psi_{++}(\beta)}{\psi_{++}(\beta)}}  &= \frac{\mathcal{Z}_\text{e-e} }{\mathcal{Z}_\text{e-e} + \mathcal{Z}_\text{e-o} + \mathcal{Z}_\text{o-e} + \mathcal{Z}_\text{o-o}} \geq 1 - e^{-\mathcal{O}(L)} ,
\end{align}
where $\mathcal{Z}_{\text{e-e}}$ is the partition function of the Ising ferromagnet at inverse temperature $\beta$ with periodic (even) boundary conditions (BC), $\mathcal{Z}_{\text{e-o}}$ has periodic (even) BC along $x$ and anti-periodic (odd) BC along $y$, and so on. We bound the middle expression by noting that every contribution to the odd sectors contains a system-spanning polymer, and consequently the odd-sector partition functions are all exponentially suppressed relative to $\mathcal{Z}_{\text{e-e}}$ at large $\beta$.

As a direct consquence, the states $\ket{\psi_{00}(\beta)}$ and $\ket{\psi_{++}(\beta)}$ are essentially indistinguishable from each other in the thermodynamic limit: the difference in expectation values of any observable is upper-bounded by $e^{-\mathcal{O}(L)}$. In particular, the original state $\ket{\psi_{++}(\beta)}$ is 1-form symmetric about non-contractible loops up to exponentially small corrections, and our parent Hamiltonian $H(\beta)$ is an excellent \emph{approximate} parent Hamiltonian for the no-go state $\ket{\psi_{00}(\beta)}$. 

Nevertheless, it is interesting to construct an \emph{exact} parent Hamiltonian for $\ket{\psi_{00}(\beta)}$ to see how our approach relates to the rigorous result of Ref.~\onlinecite{Sahay2025FunkyIsing}. Towards this end, we write $\ket{\psi_{00}(\beta)}$ in terms of a deformation operator $\hat{\mathcal{Z}}_{00}(\beta)$ analogously to Eq.~\eqref{eq:psi_deformed_app}: 
\begin{equation} \label{eqn:def00}
    \ket{\psi_{00}(\beta)} \propto \hat{\mathcal{Z}}_{00}(\beta) \ket{0}, \qquad \hat{\mathcal{Z}}_{00}(\beta) = \sum_{\LL = \partial \mathcal{R}} e^{-\beta\abs{\LL}} X_{\LL},
\end{equation}
where the sum is restricted to contractible closed loops $\mathcal{L}$, which can always be represented as the boundary of a collection of plaquettes $\mathcal{R}$. As in the previous case, $\hat{\mathcal{Z}}_{00}(\beta)$ is a positive-definite operator\footnote{To see that $\hat{\mathcal{Z}}_{00}(\beta)$ is positive-definite, we take a similar strategy as for $\hat{\mathcal{Z}}(\beta)$ in Eq.~\eqref{eq:Vdef1}. Specifically, we can express $\hat{\mathcal{Z}}_{00}(\beta)$ as the high-temperature expansion of a bond-disordered Ising model summed over periodic and antiperiodic boundary conditions:
\begin{equation}
    \hat{\mathcal{Z}}_{00}(\beta) \propto \sum_{\qty{s_v = \pm 1}} \sum_{\eta_x, \eta_y = \pm 1} \exp \qty{ K \sum_{e = \expval{v v'}} U_e^{(\eta_x, \eta_y)} X_{e} s_v s_{v'} } ,
\end{equation}
where $U_e^{(\eta_x, \eta_y)}$ equals $+1$ everywhere besides along two non-contractible loops $\tilde{\mathcal{L}}_x, \tilde{\mathcal{L}}_y$ in the dual lattice; along $\tilde{\mathcal{L}}_x$ ($\tilde{\mathcal{L}}_y$), $U_e$ instead equals $\eta_x$ ($\eta_y$). The sum over $\eta_x, \eta_y$ effectively imposes a sum over periodic and antiperiodic boundary conditions, which eliminates non-contractible cycles from the high-temperature expansion.}, and its logarithm is therefore well-defined. Our goal is to use the Mayer cluster expansion to explicitly take the logarithm of $\hat{\mathcal{Z}}_{00}(\beta)$, thereby deriving a parent Hamiltonian by the same strategy as in Appendix~\ref{app:locality_proof}. Explicitly, we will show similar to the `++' case that the logarithm admits an expansion of the form
\begin{equation}
    \hat{W}_{00}(\beta) \equiv \log \hat{\mathcal{Z}}_{00}(\beta) = \sum_{\mathcal{C}} f(\mathcal{C}) X_{\mathcal{C}}, \quad f(\mathcal{C}) \equiv \varphi(\gamma_1 , \ldots , \gamma_n) \prod_{i = 1}^n w(\gamma_i) ,
\end{equation}
in perfect analogy to Eq.~\eqref{eq:cluster_expansion_op}. The only difference is that the collection of allowed polymers $\gamma \in \Gamma$ and the form of their interactions $\delta(\gamma, \gamma')$ needs to be slightly modified relative to the `++' case, so as to disallow non-contractible loop configurations. Once this modification is made, we obtain a parent Hamiltonian for $\ket{\psi_{00}(\beta)}$ of an identical form to that of $\ket{\psi_{++}(\beta)}$ [see Eq.~\eqref{eq:sm_parent_ham}]:
\begin{equation}
\label{eq:H00_app}
    H_{00}(\beta) \equiv \sum_{e} \qty( e^{-2 \sum_{e \in \mathcal{S}_e} f(\mathcal{C}) X_{\mathcal{C}}} - Z_e ) .
\end{equation}
Finally, we will use similar techniques as in Appendices~\ref{subsec:locality_proof} and \ref{sec:proof_of_eq13} to obtain a bound on the cluster weights $f(\mathcal{C})$, which will in turn allow us to demonstrate the locality of $H_{00}(\beta)$ for any \emph{fixed} aspect ratio $a_{xy} \equiv L_x / L_y$.

\subsection{Cluster Expansion}
To employ the cluster expansion for $\hat{\mathcal{Z}}_{00}(\beta)$, our first goal is to express it as a polymer model. We immediately see that we cannot use the same set of polymers $\Gamma$ as for the $++$ case, which allow for homologically nontrivial polymers in the partition function. Instead, we define the set of polymers $\Gamma$ to contain:
\begin{enumerate}
    \item The set of connected, closed (possibly self-intersecting) loops with trivial homology, and
    \item \emph{Pairs} of connected, closed (possibly self-intersecting) loops with \emph{non-trivial} homology, such that they are disconnected from each other. 
\end{enumerate}
For brevity, we will call these ``type-1'' and ``type-2'' polymers respectively. The disconnected condition in the second case implies that both loops must be in the same homology class (i.e., they are either both contractible or both non-contractible), and therefore are together the boundary of some region $\mathcal{R}$ (note that if they were connected, then the two loops are regarded as a single type-1 polymer). The associated weights are defined the same way as before,
\begin{equation}
    \hat{w}(\gamma) = e^{-\beta \abs{\gamma}} X_{\gamma}, \qquad X_\gamma \equiv \prod_{e \in \gamma} X_e,
\end{equation}
but the interaction is modified. We define 
\begin{equation}\label{eqn:deltadef00}
    \delta(\gamma, \gamma') = \begin{cases}
        0, & \gamma, \gamma' \text{ intersect,} \\
        0, & \gamma \text{ and } \gamma' \text{ are type-2, and are disconnected and \emph{unlinked}}, \\
        1, & \text{otherwise.}
    \end{cases}
\end{equation}
We say two disconnected type-2 polymers $\gamma, \gamma'$ are \emph{linked} if, for any $\mathcal{R},\mathcal{R}'$ such that $\gamma = \partial \mathcal{R}$ and$\gamma' = \partial\mathcal{R}'$, then $\mathcal{R}$ and $\mathcal{R'}$ have overlapping plaquettes. They are \emph{unlinked} otherwise. 

To justify this choice of polymers and interactions, we must show that writing out the polymer partition function~\eqref{eqn:PolymerPartitionFunction} produces the correct deformation operator
\begin{equation}
\label{eq:app_Z00_polymer}
    \hat{\mathcal{Z}}_{00}(\beta) = \sum_{\mathcal{L} = \partial \mathcal{R}} e^{-\beta \LL} X_{\LL} = \sum_{\Gamma'\subseteq \Gamma} \qty[ \prod_{\gamma \in \Gamma'} \hat{w}(\gamma) ] \qty[ \prod_{\curly{\gamma,\gamma'} \subseteq \Gamma'} \delta(\gamma, \gamma') ].
\end{equation}
It is clear from the preceding discussion that the right-hand side produces only contractible loop configurations. Therefore, we only need to ensure that for any $\LL$ on the left-hand side, the right-hand side produces the same term with the correct coefficient (i.e., without overcounting). To this end, we make the following observations:
\begin{figure}
    \centering
    \includegraphics[width=0.5\linewidth]{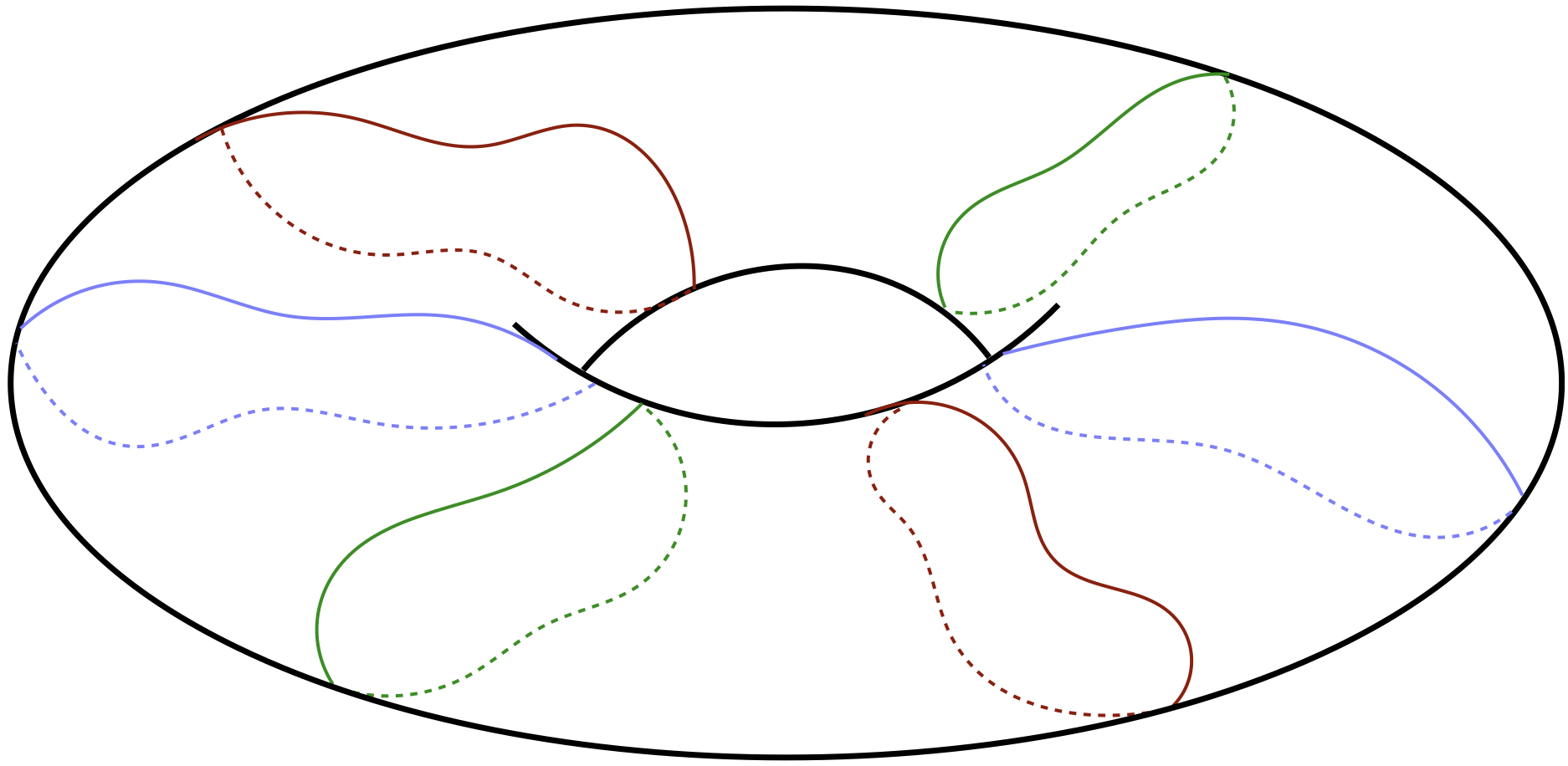}
    \caption{A loop configuration with three type-2 polymers, in maroon, blue, and green respectively. These type-2 polymers are each pairwise compatible, since they satisfy the linking criterion [see Eq.~(\ref{eqn:deltadef00})] with $\delta(\gamma, \gamma') = 1$.}
    \label{fig:Type2Polymer}
\end{figure}
\begin{enumerate}
    \item Given a loop configuration $\mathcal{L}$, there is a unique decomposition into connected loops such that all connected loops are mutually disconnected.
    \item For a given $\mathcal{L}$, if there are no non-contractible connected loops in this decomposition, then $\mathcal{L}$ is reproduced exactly by a subset $\Gamma' \subset \Gamma$ on the right-hand side containing only type-1 polymers, where each polymer corresponds to a connected loop in $\mathcal{L}$. 
    \item If the decomposition of $\mathcal{L}$ contains two non-contractible connected loops wrapping the same cycle of the torus (a single non-contractible loop is disallowed), $\mathcal{L}$ is reproduced by a subset $\Gamma'$ containing a single type-two polymer.
    \item The nontrivial case is when $\mathcal{L}$ is decomposed into four or more non-contractible connected loops, each wrapping the same cycle of the torus (Fig.~\ref{fig:Type2Polymer}). Naively, if $\mathcal{L}$ contains $2n$ such loops, there are $(2n-1)!!$ ways to pair these loops into type-2 polymers. The middle condition in the interaction~\eqref{eqn:deltadef00} ensures that only one of these pairings gives a nonzero contribution to $\hat{\mathcal{Z}}_{00}$. Specifically, a subset $\Gamma'$ containing multiple type-2 polymers contributes to $\hat{\mathcal{Z}}_{00}$ only if all type-2 polymers in $\Gamma'$ are \emph{pairwise linked}. The only way to achieve this pairwise linking is to pair the loops \emph{antipodally}; that is, if the non-contractible connected loops are labeled $1, \ldots , 2n$ in order, only the pairing
    \begin{equation}
        \qty{ \qty{1, n+1}, \qty{2, n+2}, \ldots , \qty{n, 2n} }
    \end{equation}
    into type-2 polymers leads to a nonvanishing contribution to $\hat{\mathcal{Z}}_{00}(\beta)$ (see Fig.~\ref{fig:Type2Polymer}).
\end{enumerate}
In this way, we see that the preceding definition of polymers together with the interaction~\eqref{eqn:deltadef00} avoids overcounting in the polymer partition function on the right-hand side of Eq.~\eqref{eq:app_Z00_polymer}, reproducing $\hat{\mathcal{Z}}_{00}(\beta)$ exactly.

We briefly note that the polymer model for the no-go state differs from that of the $++$ state only by the properties of non-contractible loops; the weights and interactions for type-1 polymers are identical in both partition functions. In the large-$\beta$ regime, we therefore expect that the differences between the two polymer models will be largely inconsequential in the thermodynamic limit, consistent with the observations of the preceding section and with our final conclusions. However, for purposes of constructing an exact parent Hamiltonian for any finite system size, we will exactly keep track of these exponentially small contributions from non-contractible loops.

\subsection{Modified Cluster Expansion Bound}
Since the definition of our polymer model has changed slightly, the technical details required to prove the cluster expansion bound in Eq.~\eqref{eqn:fCbound} of the main text are slightly modified. The most important qualitative modification is that, whereas Eq.~\eqref{eqn:fCbound} held at sufficiently large $\beta$ for any system size, the same bound holds in the present setting only for sufficiently large $L_x, L_y$, and at a \emph{fixed} aspect ratio $a_{xy} \equiv L_x / L_y$. While the aspect ratio can be chosen arbitrarily, it is crucial that it is held fixed as $L_x, L_y$ are taken large; as we will see, our parent Hamiltonian construction can fail if we take (for example) $L_x$ to scale as $e^{\mathcal{O}(L_y)}$, consistent with Theorem 1 of Ref.~\onlinecite{Sahay2025FunkyIsing}.

For notational simplicity, we will assume throughout $L_x \geq L_y$ without loss of generality.

\begin{lemma}[Kotecky-Preiss inequality for the polymer model $\hat{\mathcal{Z}}_{00}$]
\label{lemma:KP00}
When $\beta \geq \beta^\ast = 2 +\log 3$, and for sufficiently large system sizes $L_x, L_y$ with any fixed aspect ratio $a_{xy}$, the weight $w(\gamma) = e^{-\beta \abs{\gamma}}$ satisfies
\begin{equation}
    \sum_{\gamma' \nsim \gamma} w(\gamma') e^{|\gamma'|} \leq |\gamma|
\end{equation}
where $\gamma' \nsim \gamma$ means the two polymers are incompatible ($\delta(\gamma, \gamma') = 0$) according to the interaction~\eqref{eqn:deltadef00}. 
\end{lemma}
Note that the statement of the bound is identical to that of Lemma~\ref{lemma:KP}; the only difference lies in the collection of allowed type-2 polymers and their interactions.

\begin{proof}
    Our strategy is to divide the incompatible polymers $\gamma' \nsim \gamma$ into type-1 and type-2 polymers, respectively:
    \begin{equation}
        \begin{split}
            \sum_{\gamma' \nsim \gamma} w(\gamma') e^{|\gamma'|} &= \sum_{\substack{\gamma' \nsim \gamma \\ \gamma' \text{ type-1}}} w(\gamma') e^{|\gamma'|} + \sum_{\substack{\gamma' \nsim \gamma \\ \gamma' \text{ type-2}}} w(\gamma') e^{|\gamma'|} \\
            &= \sum_{\substack{\gamma' \nsim \gamma \\ \gamma' \text{ type-1}}} w(\gamma') e^{|\gamma'|} + \sum_{\gamma' \text{ type-2}} w(\gamma') e^{|\gamma'|} ,
        \end{split}
    \end{equation}
    where in the second line, we have extended the second sum to run over all type-2 polymers (including those compatible with $\gamma$). The sum over incompatible type-1 polymers is bounded by an identical method to that of Lemma~\ref{lemma:KP}: the number of such polymers $\gamma'$ of fixed length $|\gamma'| = s$ is upper-bounded by $|\gamma| 3^s$. 
    
    To bound the sum over type-2 polymers $\gamma'$, let $s_1$ and $s_2$ denote the respective lengths of the two connected components of $\gamma'$. Each of $s_1$ and $s_2$ can be no smaller than $L_y$, since these connected components must form non-contractible loops. For a fixed pair of lengths $(s_1, s_2)$, The total number of type-2 polymers is upper-bounded by $(L_x + L_y)^2 3^{s_1 + s_2}$; the two factors of $(L_x + L_y)$ come from observing that each of the two non-contractible loops must cross either the $x$-axis or the $y$-axis, and there are $L_x$ ($L_y$) locations at which the loop can cross the $x$-axis ($y$-axis). We therefore obtain the bound
    \begin{subequations}
    \label{eq:type2_bound}
        \begin{align}
            \sum_{\gamma' \text{ type-2}} w(\gamma') e^{|\gamma'|} &\leq (L_x + L_y)^2 \sum_{s_1, s_2 \geq L_y} 3^{s_1 + s_2} e^{-\beta(s_1 + s_2)} e^{s_1 + s_2} \\
            &\leq 4L_x^2 \qty( \frac{(3e^{1-\beta})^{L_y}}{1 - 3e^{1-\beta}} )^2 \\
            &\leq \frac{4}{(1 - e^{-1})^2} (L_x e^{-L_y})^2 e^{-2(\beta - \beta^*) L_y} ,
        \end{align}
    \end{subequations}
    valid whenever $\beta > \beta^* - 1$. Combined with the bound on type-1 polymers obtained in the proof of Lemma~\ref{lemma:KP}, we obtain the result
    \begin{subequations}
    \begin{align}
        \sum_{\gamma' \nsim \gamma} w(\gamma') e^{|\gamma'|} &\leq \frac{e^{-4}}{1 - e^{-2}} e^{-4(\beta - \beta^*)} | \gamma| + \frac{4}{(1 - e^{-1})^2} (L_x e^{-L_y})^2 e^{-2(\beta - \beta^*) L_y} \\
        &\leq |\gamma| ,
    \end{align}
    \end{subequations}
    where the last inequality holds for sufficiently large $L_y$, with $L_x = a_{xy} L_y$.
\end{proof}

Using the result of Lemma~\ref{lemma:KP00}, the analogs of Lemmas~\ref{lem:rooted} and \ref{cor:absrootedbound} are obtained by identical proofs. The analog of Lemma~\ref{lem:Cbound} is obtained similar to Lemma~\ref{lemma:KP00}, and results in a bound over cluster weights $f(\mathcal{C})$:

\begin{lemma}\label{lem:Cbound00}
    Let $e$ be an arbitrary edge on the lattice $\Lambda$. When $\beta \geq \beta^\ast = 2 + \log 3$, and for sufficiently large system sizes $L_x, L_y$ with any fixed aspect ratio $a_{xy}$, there exists $\Cclus >0 $ such that 
    \begin{equation}
        \sum_{\CC \in \mathcal{S}_e} \abs{f(\CC)} \leq \Cclus (1 + \varepsilon) e^{-4\beta} ,
    \end{equation}
    where $\varepsilon$ can be chosen arbitrarily small.
\end{lemma}
\begin{proof} 
By an identical strategy as in Lemma~\ref{lem:Cbound}, we have
\begin{subequations}
\begin{align}
    \sum_{\CC \in \mathcal{S}_e} \abs{f(\CC)} &\leq \sum_{n\geq 1} \sum_{\substack{\gamma_1,\dots,\gamma_n\\ \exists \gamma_k \ni e}} \abs{\varphi(\gamma_1,\dots,\gamma_n)} \prod_{i} w(\gamma_i) \\
    &\leq \sum_{n\geq 1} n \sum_{\gamma_1 \ni e}\sum_{\gamma_2,\dots,\gamma_n} \abs{\varphi(\gamma_1,\dots,\gamma_n)} \prod_{i} w(\gamma_i) \\
    &= \sum_{\gamma_1 \ni e} g(\gamma_1)\\
    &\leq \sum_{\gamma \ni e} w(\gamma) e^{|\gamma|} \\
    &\leq \term{ \sum_{s \geq 4, \mathrm{even}} 3^s e^{-\beta s} e^{s} + \sum_{\gamma \text{ type-2}} w(\gamma) e^{|\gamma|}}\\
    &\leq \frac{(3e^{1-\beta})^4}{1 - (3e^{1-\beta})^2} + \frac{4}{(1 - e^{-1})^2} (L_x e^{-L_y})^2 e^{-2(\beta - \beta^*) L_y} \\
    &\leq \Cclus (1 + \varepsilon) e^{-4\beta} ,
\end{align}
\end{subequations}
where we have once again broken the sum over polymers into type-1 and type-2 polymers, using the bound in Eq.~\eqref{eq:type2_bound}, and we have identified $\Cclus = [(3e)^4 / (1 - e^{-2})] \approx 5.115 \times 10^3$. The arbitrarily small parameter $\varepsilon$ in the last expression is used to bound the second term in the previous line; for any fixed $\varepsilon > 0$, $L_y$ and $L_x = a_{xy} L_y$ can be taken sufficiently large that the inequality is satisfied.

\end{proof}

\subsection{Proving that \texorpdfstring{$H_{00}(\beta)$}{H00beta} is local and gapped}
\label{subsec:H00_locality}
Finally, we use the result of Lemma~\ref{lem:Cbound00} to prove that the parent Hamiltonian $H_{00}(\beta)$ is local in the sense of Eqs.~\eqref{eq:hastsings_1}, \eqref{eq:hastings_2}, and \eqref{eq:hastings_3}  at sufficiently large $\beta$, from which Theorem 2 implies that $H_{00}(\beta)$ exhibits a nonzero spectral gap. Our approach mirrors that of the `++' case in Appendix~\ref{app:locality_proof}: once again, we Taylor expand the exponential in Eq.~\eqref{eq:H00_app} (justified by Lemma~\ref{lem:Cbound00}) to write
\begin{equation}
\label{eq:H00_app_2}
    H_{00}(\beta) = H_0 + \sum_e \sum_{n \geq 1} \frac{(-2)^n}{n!} \sum_{\mathcal{C}_1 \ldots \mathcal{C}_n \in \mathcal{S}_e} \prod_{j = 1}^n \qty[ f(\mathcal{C}_j) X_{\mathcal{C}_j} ] .
\end{equation}
To divide the sum into local terms, it is convenient to distinguish between clusters $\mathcal{C}$ containing entirely type-1 polymers and clusters which contain some type-2 polymers. We call a cluster $\mathcal{C} = (\gamma_1, \ldots , \gamma_n)$ type-1 if all of its constituent polymers are all type-1; otherwise, we call $\mathcal{C}$ type-2. We then group terms in Eq.~\eqref{eq:H00_app_2} based on whether they contain type-2 clusters or not. Specifically, we first define
\begin{equation}
    V_{r, v} \equiv \sum_{e = (v, v + \hat{x}), (v, v + \hat{y})} \sum_{n \geq 1} \frac{(-2)^n}{n!} \sum_{\substack{\mathcal{C}_1 \ldots \mathcal{C}_n \in \mathcal{S}_e \\ \mathcal{C}_j \text{ type 1 } \forall j \\ \max_j \norm{\mathcal{C}_j} = r + 1}} \prod_{j = 1}^n [f(\mathcal{C}_j) X_{\mathcal{C}_j}] ,
\end{equation}
for all vertices $v$ and all odd $r \geq 3$, in perfect analogy to Eq.~\eqref{eq:V_rv_app}, save for the restriction to exclusively type-1 clusters. Recall that the restriction to clusters with maximum length $\max_j \norm{\CC_j} = r+1$ implies that $V_{r,v}$ is contained within the $r \times r$ square $A_v(r) \in S(r)$ centered on vertex $v$. To capture terms with type-2 clusters, we also define the (highly nonlocal) operators
\begin{equation}
    V_v' \equiv \sum_{e = (v, v + \hat{x}), (v, v + \hat{y})} \sum_{n \geq 1} \frac{(-2)^n}{n!} \sum_{\substack{\mathcal{C}_1 \ldots \mathcal{C}_n \in \mathcal{S}_e \\ \exists \text{ type-2 } \mathcal{C}_j}} \prod_{j = 1}^n [f(\mathcal{C}_j) X_{\mathcal{C}_j}] .
\end{equation}
Combining these two terms, we obtain a partition of $H_{00}(\beta)$ into terms as follows:
\begin{equation}
    H_{00}(\beta) = \sum_{v} \sum_{r \geq 3, \text{ odd}} V_{r, v} + \sum_v V'_v .
\end{equation}

By an identical calculation to that of Eqs.~\eqref{eq:vbound_app} and \eqref{eq:vbound_app_2}, we immediately obtain the same bound on the spectral norm of $V_{r,v}$ in the present setting:
\begin{equation}
    \norm{V_{r,v}} \leq J(\beta) e^{-r}, \quad J(\beta) \equiv (1 + \varepsilon) \tilde{J} e^{-4\beta} ,
\end{equation}
where $\tilde{J}$ is given explicitly in Eq.~\eqref{eq:vbound_app_final}, and the parameter $\varepsilon$ can be chosen to be arbitrarily small (see Lemma~\ref{lem:Cbound00}). Therefore, by Theorem~\ref{theorem:BHM}, the \emph{modified} Hamiltonian
\begin{equation}
    \widetilde{H}_{00}(\beta) \equiv \sum_{v} \sum_{r \geq 3, \text{ odd}} V_{r, v}
\end{equation}
is gapped and nondegenerate at sufficiently large $\beta$ and sufficiently large system sizes. On the other hand, we will now show that the remainder term $\sum_v V'_v$ exhibits an exponentially small spectral norm, and therefore the exact parent Hamiltonian $H_{00}(\beta)$ is also gapped and nondegnerate in this regime. Explicitly, by a calculation similar to Eq.~\eqref{eq:vbound_app}, we have
\begin{subequations}
    \begin{align}
        \norm{V_v'} &\leq \sum_{e = (v, v + \hat{x}), (v, v + \hat{y})} \sum_{n \geq 1} \frac{2^n}{n!} \sum_{\substack{\CC_1 \ldots \CC_n \in \mathcal{S}_e \\ \exists \text{ type-2 } \CC_j}} \prod_{j = 1}^n \abs{f(\CC_j)} \\
        &\leq \sum_{e = (v, v + \hat{x}), (v, v + \hat{y})} \sum_{n \geq 1} \frac{2^n}{(n-1)!} \qty( \sum_{\substack{\CC \in \mathcal{S}_e \\ \CC \text{ type-2}}} \abs{f(\CC)} ) \qty( \sum_{\CC \in \mathcal{S}_e} \abs{f(\CC)} )^{n-1}
    \end{align}
\end{subequations}
Analogously to Eq.~\eqref{eq:vbound_app_2}, we bound the last sum over clusters using the result of Lemma~\ref{lem:Cbound00}, and we bound the sum over type-2 clusters by noting that $f(\CC) \leq e^{-\beta L_y} f_{\beta / 2}(\CC)$, since each type-2 cluster $\CC$ has cluster length $\norm{\CC} \geq 2 L_y$. Following the same steps as in Eq.~\eqref{eq:vbound_app_2}, we obtain the bound
\begin{equation}
    \norm{V'_v} \leq C_{\text{global}} e^{-\beta L_y}, \quad C_{\text{global}} \equiv 4 \Cclus e^{-2 \beta^*} \exp \qty[ 2 \Cclus e^{-4 \beta^*} ].
\end{equation}
Correspondingly, $\norm{\sum_v V_v'} \leq L_x L_y C_{\text{global}} e^{-\beta L_y}$ is exponentially small for sufficiently large $L_x, L_y$, so long as the aspect ratio $a_{xy}$ is held fixed. By Weyl's inequality, the spectral gap of $H_{00}(\beta)$ cannot differ from that of $\widetilde{H}_{00}(\beta)$ by more than twice the spectral norm of $\sum_v V_v'$, and therefore $H_{00}(\beta)$ also exhibits a nonzero spectral gap in the thermodynamic limit for sufficiently large $\beta$.

\section{Equivalent Definitions of locality}
\label{app:locality_equivalence}
In this Appendix, we show that the exponential-in-diameter notion of locality [as defined in Eq.~(\ref{eq:locality_bound}) of the main text] and the notion of locality used to prove the stability of gap theorem in Ref.~\onlinecite{BravyiHastingsMichalakis2010} [Eq.~(\ref{eq:hastings_3})] are equivalent. For completeness, we recap the definitions here. Note that the \emph{exponential-in-volume} definition of locality [Eq.~(\ref{eq:expvol_locality})] is a strictly stronger criterion than the definitions discussed in this section.

\begin{defn}[Exponential-in-diameter locality]
    A Hamiltonian $H = \sum_{\mathcal{R} \subseteq \Lambda} h_{\mathcal{R}}$, where each local term $h_{\mathcal{R}}$ is supported on subregion $\mathcal{R}$, is considered \emph{exponential-in-diameter local} if there exists finite constants $\mu,s> 0$ such that 
    \begin{equation}
        \sup_{v \in \Lambda} \sum_{\mathcal{R} \ni v} \norm{h_{\mathcal{R}}} \abs{\mathcal{R}} e^{\mu \diam(\mathcal{R})} \leq s.
    \end{equation}
\end{defn}

\begin{defn}[Bravyi-Hastings-Michalakis (BHM) locality \cite{BravyiHastingsMichalakis2010}]
    We write the Hamiltonian as $H = \sum_{r \geq 1} \sum_{A \in S(r)} V_{r,A}$, where $S(r)$ is the set of all $r\times r$ squares on the lattice, and $V_{r,A}$ has support on the $r\times r$ square $A$. Such a Hamiltonian is \emph{BHM-local} if there exists constants $J,\mu' > 0 $ such that
\begin{equation}
    \sup_{A\in S(r)} \norm{V_{r,A}} \leq J e^{-\mu' r} .
\end{equation}
We will call $J$ the strength of the Hamiltonian and $\mu'$ the decay exponent. 
\end{defn}

\subsection{Definition 1 \texorpdfstring{$\implies$}{implies} Definition 2}
Given an exponential-in-diameter local Hamiltonian $H = \sum_{\mathcal{R} \subseteq \Lambda} h_{\mathcal{R}}$, let us define the local terms
\begin{equation}
    V_{r, v} \equiv \sum_{\substack{\mathcal{R} \ni v \\ \diam(\mathcal{R}) = r}} \frac{1}{\abs{\mathcal{R}}} h_{\mathcal{R}} .
\end{equation}
Note that each $h_{\mathcal{R}}$ contributes to $\abs{\mathcal{R}}$ distinct terms $V_{r,v}$, and therefore $H = \sum_{v \in \Lambda} \sum_{r \geq 1} V_{r,v}$. Moreover, $V_{r,v}$ is supported on the $(2r-1) \times (2r-1)$ square centerd at $v$. Finally, its norm is upper-bounded as
\begin{subequations}
    \begin{align}
        \norm{V_{r,v}} &\leq \sum_{\substack{\mathcal{R} \ni v \\ \diam(\mathcal{R}) = r}} \frac{1}{|\mathcal{R}|} \norm{h_{\mathcal{R}}} \\
        &= e^{-\mu r} \sum_{\substack{\mathcal{R} \ni v \\ \diam(\mathcal{R}) = r}} \frac{|\mathcal{R}|}{|\mathcal{R}|^2} \norm{h_{\mathcal{R}}} e^{\mu \diam(\mathcal{R})} \\
        &\leq s e^{-\mu r} \\
        &= s e^{-\mu / 2} e^{- \frac{\mu}{2}(2r-1)} .
    \end{align}
\end{subequations}
We therefore find that exponential-in-diameter locality implies BHM locality, with strength $J = s e^{-\mu/2}$ and decay exponent $\mu' = \mu/2$.

\subsection{Definition 2 \texorpdfstring{$\implies$}{implies} Definition 1}
Given a BHM-local Hamiltonian $H = \sum_{r \geq 1} \sum_{A \in S(r)} V_{r,A}$, we simply identify $h_{\mathcal{R}} = V_{r, A}$ with $\mathcal{R} = A$. Then, letting $\mu$ be an unspecified constant for now, we have
\begin{subequations}
\begin{align}
    \sum_{\mathcal{R} \ni v} \norm{h_{\mathcal{R}}} \abs{\mathcal{R}} e^{\mu \diam(\mathcal{R})} &= \sum_{r\geq 1} \sum_{\substack{A \in S(r) \\ A \ni v}} \norm{V_{r,A}} \abs{A} e^{\mu \diam(A)} \\
    &= \sum_{r\geq 1} \sum_{\substack{A \in S(r) \\ A \ni v}} \norm{V_{r,A}} r^2 e^{\mu\sqrt{2} r} \\
    &\leq \sum_{r\geq 1} r^2 J e^{-\mu' r} r^2 e^{\mu\sqrt{2} r} ,
\end{align}
\end{subequations}
where in the final line, we've noted that there are $r^2$ squares $A \in S(r)$ containing a given vertex $v$. The final sum converges as long as $\mu < \mu' / \sqrt{2}$. We therefore find that BHM locality implies exponential-in-diameter locality, for some decay constant $\mu < \mu' / \sqrt{2}$ and strength $s \propto J$ given explicitly by the above sum.

\section{Exponential-in-volume locality bounds}
\label{sec:exp_in_vol}
In this Appendix, we show that our parent Hamiltonian $H(\beta)$ not only satisfies the exponential-in-diameter and BHM locality conditions described in Appendix~\ref{app:locality_equivalence}, but also satisfies a stronger \emph{exponential-in-volume} locality condition:
\begin{defn}[Exponential-in-volume locality]
    A Hamiltonian $H = \sum_{\mathcal{R} \subseteq \Lambda} h_{\mathcal{R}}$, where each local term $h_{\mathcal{R}}$ is supported on the \emph{connected} subregion $\mathcal{R}$, is considered \emph{exponential-in-volume} local if there exists finite constants $\mu, s > 0$ such that
    \begin{equation}
    \label{eq:expvol_locality}
        \sup_v \sum_{\mathcal{R} \ni v} \norm{h_{\mathcal{R}}} |\mathcal{R}| e^{\mu |\mathcal{R}|} \leq s .
    \end{equation}
    Note that this condition effectively requires $\norm{h_{\mathcal{R}}}$ to decay exponentially with the \emph{volume} $|\mathcal{R}|$ of the subregion $\mathcal{R}$. This is a stronger condition than the exponential-in-diameter decay required of the preceding section.
\end{defn}
From the explicit form of our parent Hamiltonian $H(\beta)$ [see Eq.~\eqref{eq:sm_parent_ham}], it is quite natural that $H(\beta)$ satisfies such a bound: indeed, the nontrivial component of $H(\beta)$ consists of a sum of loop operators exponentially suppressed in their length $\mathcal{L}$, which is generally proportional to both the loop's diameter and the volume of its support. Furthermore, using the results of Ref.~\onlinecite{Yin_Lucas_low-density_2025, DeRoeck_Khemani_low-density_2025}, the existence of a gapped exponential-in-volume local parent Hamiltonian $H(\beta)$ for $\ket{\psi(\beta)}$ implies that $\ket{\psi(\beta)}$ can be prepared (up to an error exponentially small in the system size) by finite-time unitary evolution using an exponential-in-volume local quasiadiabatic generator. 

To prove that $H(\beta) = H_0 + V$ satisfies exponential-in-volume locality, it is sufficient to focus on the nontrivial component $V$, which we once again expand in Taylor series:
\begin{equation}
    V \equiv \sum_e \qty( e^{-2 \sum_{\CC \in \mathcal{S}_e} f(\CC) X_{\CC}} - 1 ) = \sum_e \sum_{n \geq 1} \frac{(-2)^n}{n!} \sum_{\CC_1 \ldots \CC_n \in \mathcal{S}_e} \prod_{j = 1}^n [f(\CC_j) X_{\CC_j}] .
\end{equation}
Our first goal is to group $V$ into local terms. Towards this end, we associate a connected region $\mathcal{R}(\CC_1 , \ldots , \CC_n)$ to each collection of connected clusters $(\CC_1, \ldots , \CC_n)$ with $\CC_j \equiv (\gamma^{(j)}_1 , \ldots , \gamma^{(j)}_{m_j})$ as follows:
\begin{equation}
    \mathcal{R}(\CC_1 , \ldots , \CC_n) \equiv \bigcup_{j} \qty(  \gamma^{(j)}_1 \cup \ldots \cup \gamma^{(j)}_{m_j} ) .
\end{equation}
Such regions $\mathcal{R}$ appearing in $V$ are guaranteed to be connected, since all of the clusters $\CC_j$ comprising a given region are required to contain the same edge $e$. Note that the same region will generally appear several times within $V$; by treating them as separate in computing the sum~\eqref{eq:expvol_locality}, we are strictly overestimating the left-hand side, since we are ignoring potential sign cancellations. We therefore have the inequality
\begin{equation}
    \begin{split}
        \sup_e \sum_{\mathcal{R} \ni e} \norm{V_{\mathcal{R}}} |\mathcal{R}| e^{\mu |\mathcal{R}|} &\leq \sup_e \sum_{e'} \sum_{n \geq 1} \sum_{\substack{\mathcal{C}_1 \ldots \mathcal{C}_n \in \mathcal{S}_{e'} \\ \mathcal{R}(\mathcal{C}_1 , \ldots , \mathcal{C}_n) \ni e}} \frac{2^n}{n!} \abs{f(\CC_1)} \ldots \abs{f(\CC_n)} \abs{\mathcal{R}(\mathcal{C}_1 , \ldots , \mathcal{C}_n)} e^{\mu \abs{\mathcal{R}(\mathcal{C}_1 , \ldots , \mathcal{C}_n)}} 
    \end{split}
\end{equation}
Due to translation invariance, the supremum over $e$ can be discarded. In fact, we can simplify the sum by \emph{averaging} over all $e$ as follows:
\begin{subequations}
    \begin{align}
        \sum_{e'} \sum_{n \geq 1} \sum_{\substack{\mathcal{C}_1 \ldots \mathcal{C}_n \in \mathcal{S}_{e'} \\ \mathcal{R}(\mathcal{C}_1 , \ldots , \mathcal{C}_n) \ni e}} \qty(\cdots) &= \frac{1}{N_e} \sum_{e,e'} \sum_{n \geq 1} \sum_{\substack{\mathcal{C}_1 \ldots \mathcal{C}_n \in \mathcal{S}_{e'} \\ \mathcal{R}(\mathcal{C}_1 , \ldots , \mathcal{C}_n) \ni e}} (\cdots) \\
        &= \frac{1}{N_e} \sum_{e'} \sum_{n \geq 1} \sum_{\mathcal{C}_1 \ldots \mathcal{C}_n \in \mathcal{S}_{e'}} \abs{\mathcal{R}(\CC_1 , \ldots , \CC_n)} (\cdots) \\
        &= \sum_{n \geq 1} \sum_{\mathcal{C}_1 \ldots \mathcal{C}_n \in \mathcal{S}_{e}} \abs{\mathcal{R}(\CC_1 , \ldots , \CC_n)} (\cdots) ,
    \end{align}
\end{subequations}
where in the second line we've performed the sum over $e$ by noting that there are precisely $\abs{\mathcal{R}(\CC_1 , \ldots , \CC_n)}$ nonvanishing terms in the sum, each with the same value; and in the third line we've noted that every term in the sum over $e'$ yields the same value, allowing us to replace the sum over $e'$ with a fixed edge $e$. We therefore have the inequality
\begin{subequations}
    \begin{align}
        \sum_{\mathcal{R} \ni e} \norm{V_{\mathcal{R}}} |\mathcal{R}| e^{\mu |\mathcal{R}|} &\leq \sum_{n \geq 1} \sum_{\CC_1 \ldots \CC_n \in \mathcal{S}_e} \frac{2^n}{n!} \abs{f(\CC_1)} \ldots \abs{f(\CC_n)} \abs{ \mathcal{R}(\CC_1 , \ldots , \CC_n) }^2 e^{\mu \abs{\mathcal{R}(\CC_1 , \ldots , \CC_n)}} \\
        &\leq \sum_{n \geq 1} \sum_{\CC_1 \ldots \CC_n \in \mathcal{S}_e} \frac{2^n}{n!} \abs{f(\CC_1)} \ldots \abs{f(\CC_n)}  \Big[ \norm{\CC_1} + \ldots + \norm{\CC_n} \Big]^2 e^{\mu \big[\norm{\CC_1} + \ldots + \norm{\CC_n} \big] } \\
        &\leq \sum_{n \geq 1} \sum_{\CC_1 \ldots \CC_n \in \mathcal{S}_e} \frac{2^n}{n!} \abs{f(\CC_1)} \ldots \abs{f(\CC_n)} e^{\big[ \norm{\CC_1} + \ldots + \norm{\CC_n} \big]} e^{\mu \big[ \norm{\CC_1} + \ldots + \norm{\CC_n} \big]} \\
        &= \exp \qty{ 2 \sum_{\CC \in \mathcal{S}_e} \abs{f(\CC)} e^{(\mu + 1) \norm{\CC}} } - 1,
    \end{align}
\end{subequations}
recalling that $\norm{\CC} = \sum_{\gamma \in \CC} \abs{\gamma}$ is the total length of all polymers in the cluster $\CC$. In the third line we have used the inequality $x^2 \leq e^{x}$ for all $x > 0$, applied to $x = \norm{\CC_1} + \ldots + \norm{\CC_n}$.

As in the proofs of Sec.~\ref{subsec:locality_proof}, we bound the sum over clusters by writing $f(\CC) = e^{-\beta \norm{\CC} / 2} f_{\beta / 2}(\CC)$, where $f_{\beta / 2}(\CC)$ is the same function evaluated at deformation strength $\beta / 2$. As long as $\beta > 2 (\mu + 1)$, we then have
\begin{equation}
    \begin{split}
        \sum_{\mathcal{R} \ni e} \norm{V_{\mathcal{R}}} |\mathcal{R}| e^{\mu |\mathcal{R}|} &\leq \exp \qty{ 2 \sum_{\CC \in \mathcal{S}_e} \abs{f_{\beta / 2}(\CC)} e^{(\mu + 1 - \beta/2) \norm{\CC}} } - 1 \\
        &\leq \exp \qty{ 2 \sum_{\CC \in \mathcal{S}_e} \abs{f_{\beta / 2}(\CC)} } - 1 \\
        &\leq \exp \qty{ 2 \Cclus e^{-2\beta} } - 1 ,
    \end{split}
\end{equation}
where in the final line, we have used the bound in Eq.~\eqref{eqn:fCbound} of the main text (proven in Lemma~\ref{lem:Cbound}), valid for $\beta > 2\beta^* = 2(2 + \log 3)$. Note that the final expression can be made arbitrarily small at large $\beta$.

Altogether, we find that for \emph{any} choice of the constants $\mu, s$, the perturbation $V$ satisfies the exponential-in-volume locality criterion in Eq.~\eqref{eq:expvol_locality} at sufficiently large values of $\beta$. A similar locality bound can be immediately obtained for the full parent Hamiltonian $H(\beta) = H_0 + V$ by adding $e^{\mu}$ to the value of $s$, to account for the contribution from the strictly local term in $H_0$.

\section{The Strongly Pauli-\texorpdfstring{$Z$}{Z} Decohered Toric Code is Short-Range Entangled}
In this section, we outline a rigorous proof that the 2D toric code under sufficiently large Pauli-$Z$ decoherence strength is separable as a convex sum of trivial pure states. To do so, we exploit the decomposition of the decohered toric code into pure states given in Refs.~\onlinecite{ChenGrover_separability_2024,WangSongMengGrover_analog_2025}, and utilize our construction to construct gapped parent Hamiltonians for these states. 

\subsection{Review: Separability Transition of the Decohered Toric Code}
For completeness, we first review the main result of Ref.~\onlinecite{ChenGrover_separability_2024}, which constructs a particular pure-state decomposition for the toric code subjected to a single species of Pauli errors. Readers who are already familiar with the result of Ref.~\onlinecite{ChenGrover_separability_2024} can skip to the next section.

We consider the toric code subjected to incoherent Pauli-$Z$ decoherence, such that the local error channel on edge $e$ is given by 
\begin{equation}
    \mathcal{E}_e(\rho) = (1-p) \rho + p Z_e \rho Z_e.
\end{equation}
Applying the same channel independently to all edges, we have 
\begin{equation}
    \begin{split}
        \rho_p \equiv \mathcal{E}(\dyad{\TC}) &= \sum_E (1-p)^{N_e - \abs{E}} p^{\abs{E}} Z_E \dyad{\TC} Z_E \\
        &= (1-p)^{N_e} \sum_E \term{\frac{p}{1-p}}^{\abs{E}}  Z_E \dyad{\TC} Z_E,
    \end{split}
\end{equation}
where $\mathcal{E} \equiv \prod_e \mathcal{E}_e$, and $E$ is the `error configuration', i.e., the collection of edges which were acted on nontrivially by the error, so that $Z_E \equiv \prod_{e\in E} Z_e$ is the corresponding error.

To simplify the decohered state, first note that $Z_E \ket{\TC} = Z_{E'} \ket{\TC}$ whenever $Z_E$ and $Z_{E'}$ differ by vertex stabilizers $A_v = \prod_{e \ni v} Z_e$. We can therefore sort errors into equivalence classes, denoted $[E]$ with representative $E$, and reorganize the sum over errors as a sum over equivalence classes:
\begin{equation}
    \begin{split}
        \rho_p &\propto \sum_{[E]} \qty[ \sum_{E' \sim E} \qty( \frac{p}{1-p} )^{|E'|}  ] Z_E \dyad{\TC} Z_E \\
        &= \sum_{[E]} \mathcal{Z}_p(E) \, Z_E \dyad{\TC} Z_E ,
    \end{split}
\end{equation}
where $E' \sim E$ means that $Z_E$ and $Z_{E'}$ differ by products of $A_v$ stabilizers, and the relative probabilities $\mathcal{Z}_p(E)$ of each equivalence class are proportional to the partition functions of a bond-disordered Ising model \cite{dennis_topological_2002}:
\begin{equation}
    \begin{split}
        \mathcal{Z}_p(E) &\equiv \sum_{E' \sim E} \qty( \frac{p}{1-p} )^{|E'|} \\
        &\propto \mathcal{Z}^{\text{Ising}}_K(\qty{x_e}) \equiv \sum_{\qty{s_v = \pm 1}} \exp \qty{ K \sum_{e = \expval{v v'}} x_e s_v s_{v'} } ,
    \end{split}
\end{equation}
with $K = - \frac{1}{2} \log[p / (1 - p)]$, and $x_e = -1$ ($x_e = +1$) whenever $e \in E$ ($e \not\in E$). Note that $Z_p(E)$ is independent of the choice of representative $E$; in the language of the Ising model, the partition function depends only on the locations of frustrated plaquettes and the homology of the disorder lines connecting them through the dual lattice. It can be shown from the above representation of $\rho_p$ that there is a sharp phase transition at $p = p_c \approx 0.109$ in the behavior of an optimal decoder which attempts to recover $\ket{\TC}$ from $\rho_p$, whose critical phenomena is described by the random-bond Ising model on the Nishimori line~\cite{dennis_topological_2002,fan_diagnostics_2024}. 

In Ref.~\onlinecite{ChenGrover_separability_2024}, the authors argued that the phase transition in optimal decoding coincides with a \emph{separability} transition: for $p > p_c$, it becomes possible to express $\rho_p$ as a convex sum of short-range entangled (SRE) states. To demonstrate this result, let us first  consider the initial state $\ket{\TC} = \ket{\TC_{++}}$. In this case, $Z_E \ket{\TC_{++}}$ is orthogonal to $Z_{E'} \ket{\TC_{++}}$ whenever $E \nsim E'$. This allows us to trivially take the square root of $\rho_p$:
\begin{equation}
    \sqrt{\rho_p} \propto \sum_{[E]} \sqrt{\mathcal{Z}_p(E)} \, Z_E \dyad{\TC_{++}} Z_E .
\end{equation}
Following Ref.~\onlinecite{ChenGrover_separability_2024}, we write $\rho_p = \sqrt{\rho_p} \sqrt{\rho_p}$ and insert a resolution of identity in the Pauli-$Z$ basis:
\begin{equation}
    \begin{split}
        \rho_p &\propto \sum_{\qty{z_e = \pm 1}} \sqrt{\rho_p} \dyad{\qty{z_e}} \sqrt{\rho_p} \\
        &= \sum_{\mathcal{L}} X_{\mathcal{L}} \qty(\sqrt{\rho_p} \dyad{0} \sqrt{\rho_p} ) X_{\mathcal{L}} ,
    \end{split}
\end{equation}
Where we've noted that the only Pauli-$Z$ basis states having nonzero overlap with $\rho_p$ are those satisfying $A_v = +1$ (i.e., closed loop states), and therefore can be written as $X_{\mathcal{L}} \ket{0}$ for some loop operator $X_{\mathcal{L}}$; additionally, note that $X_{\mathcal{L}}$ commutes with $\rho_p$ and $\sqrt{\rho_p}$. We can further simplify the pure state in the parenthesis as follows:
\begin{subequations}
    \begin{align}
        \sqrt{\rho_p} \ket{0} &= \sum_{[E]} \sqrt{\mathcal{Z}_p(E)} \, Z_E \dyad{\TC_{++}} Z_E \ket{0} \\
        &\propto \sum_{[E]} \sqrt{\mathcal{Z}_p(E)} \, Z_E \qty( \prod_v (1 + A_v) \prod_p(1 + B_p) \prod_{\mu = 1, 2}( 1 + \overline{X}_{\mu} ) ) Z_E \ket{0} \\
        &= \sum_{[E]} \sqrt{\mathcal{Z}_p(E)} \, Z_E \qty( \prod_p(1 + B_p) \prod_{\mu = 1, 2}( 1 + \overline{X}_{\mu} ) ) Z_E \ket{0} \\
        &\propto \sum_{[E]} \sqrt{\mathcal{Z}_p(E)} \, Z_E \qty( \prod_p(1 + B_p) \prod_{\mu = 1, 2}( 1 + \overline{X}_{\mu} ) ) Z_E \sum_{E'} Z_{E'} \ket{+} ,
    \end{align}
\end{subequations}
where in the second line, $\overline{X}_{\mu} \equiv X_{\mathcal{L}_{\mu}}$ are a basis of two logical-$X$ operators, given by Wilson loop operators about two representative non-contractible loops $\mathcal{L}_{\mu}$ of the torus. In the final line, $\ket{+}$ is the simultaneous +1 eigenstate of each $X_e$, and the projector $Z_E \qty( \cdots ) Z_E$ annihilates the state to the right unless $E$ and $E'$ are homologous. Finally using that $\mathcal{Z}_p(E)$ is independent of the chosen representative, we obtain the simple expression
\begin{equation}
    \begin{split}
        \sqrt{\rho_p} \ket{0} &\propto \sum_{E} \sqrt{\mathcal{Z}_p(E)} \, Z_E \ket{+} \\
        &\propto \sum_{\qty{x_e = \pm 1}} \sqrt{\mathcal{Z}_K(\qty{x_e})} \ket{\qty{x_e}}_x ,
    \end{split}
\end{equation}
where $\ket{\qty{x_e}}_x$ are the simultaneous Pauli-$X$ eigenstates. In summary, we find that $\rho_p$ can be written as a convex sum of pure states as follows:
\begin{equation}
\label{eq:phi_p}
    \rho_p \propto \sum_{\mathcal{L}} X_{\mathcal{L}} \dyad{\phi(p)} X_{\mathcal{L}}, \quad \ket{\phi(p)} \equiv \sqrt{\rho_p} \ket{0} \propto \sum_{\qty{x_e = \pm 1}} \sqrt{\mathcal{Z}_K(\qty{x_e})} \ket{\qty{x_e}}_x. 
\end{equation}

The authors of Ref.~\onlinecite{ChenGrover_separability_2024} argued that the state $\ket{\phi(p)}$ (and thereby each state $X_{\mathcal{L}} \ket{\phi(p)}$, which differs from $\ket{\phi(p)}$ by the finite-depth unitary $X_{\mathcal{L}}$) is short-range entangled from the following properties:
\begin{enumerate}
    \item One-form symmetry: for \emph{contractible} loops $\tilde{\mathcal{L}}$ in the dual lattice, $Z_{\tilde{\mathcal{L}}} \ket{\phi(p)} = \ket{\phi(p)}$. This follows by noting that $Z_{\tilde{\mathcal{L}}}$ commutes with $\sqrt{\rho_p}$, so that $Z_{\tilde{\mathcal{L}}} \sqrt{\rho_p} \ket{0} = \sqrt{\rho_p} \ket{0}$. Alternatively, $Z_{\tilde{\mathcal{L}}}$ ads a contractible closed loop to the disorder realization $\qty{x_e}$, leaving the partition functions $\mathcal{Z}_K(\qty{x_e})$ invariant.

    \item Condensed $m$ anyons: if we instead consider an open string $\tilde{P}$ through the dual lattice, we find
    \begin{equation}
        \frac{\bra{\phi(p)} Z_{\tilde{P}} \ket{\phi(p)}}{\bra{\phi(p)} \ket{\phi(p)}} = \sum_E (1-p)^{N_e - |E|} p^{|E|} \sqrt{\frac{Z_p(E \oplus \tilde{P})}{Z_p(E)}} .
    \end{equation}
    This quantity is a two-point disorder parameter correlator in the Nishimori random-bond Ising model; it decays exponentially for $p < p_c$ and is long-range ordered for $p > p_c$. We therefore find that $\ket{\phi(p)}$ exhibits condensed $m$ anyons for $p > p_c$. 
\end{enumerate}
Subsequently, Ref.~\onlinecite{WangSongMengGrover_analog_2025} demonstrated numerically that $\ket{\phi(p)}$ exhibits zero topological entanglement entropy. It follows from these observations that $\ket{\phi(p)}$ realizes a topologically trivial confined phase. 

However, it is also interesting to note that $\ket{\phi(p)}$ exhibits \emph{perimeter-law} Wilson loop correlation functions; indeed, since $X_{\mathcal{L}}$ commutes with $\sqrt{\rho_p}$, we have
\begin{equation}
    \frac{\bra{\phi(p)} X_{\mathcal{L}} \ket{\phi(p)}}{\bra{\phi(p)} \ket{\phi(p)}} = \frac{\bra{0} X_{\mathcal{L}} \rho_p \ket{0}}{\bra{0} \rho_p \ket{0}} = \sum_E (1 -p)^{N_e - |E|} p^{|E|} \frac{\bra{0} X_{\mathcal{L}} Z_E \ket{\TC} \bra{\TC} \ket{0}}{\bra{0} \ket{\TC} \bra{\TC} \ket{0}} = (1 - 2p)^{|\mathcal{L}|} .
\end{equation}
In other words, the state $\ket{\phi(p)}$ is qualitatively similar to the deformed toric code state studied in the main text. Given the results of Ref.~\onlinecite{Sahay2025FunkyIsing}, one may naturally wonder whether these states are indeed short-range entangled (SRE) or not. In the following section, we will demonstrate explicitly that $\ket{\phi(p)}$ is SRE at sufficiently large $p$ by constructing a local gapped parent Hamiltonian for which $\ket{\phi(p)}$ is the unique ground state.

\subsection{Constructing a Parent Hamiltonian for \texorpdfstring{$\ket{\phi(p)}$}{phi_p}}
To construct a parent Hamiltonian for $\ket{\phi(p)}$, we note from Eq.~\eqref{eq:phi_p} that it can be expressed in terms of the deformation operator $\hat{\mathcal{Z}}(\beta)$ defined for the deformed toric code [see Eq.~\eqref{eq:operator_partition_fns}, as well as Eq.~\eqref{eq:Vdef1} of the main text]:
\begin{equation}
    \ket{\phi(p)} \propto \sum_{\qty{x_e = \pm 1}} \sqrt{ \hat{\mathcal{Z}}(\beta) } \ket{\qty{x_e}}_x \propto \sqrt{\hat{\mathcal{Z}}(\beta)} \ket{0} = \exp{ \frac{1}{2} \hat{W}(\beta) } \ket{0} = \exp \qty{ \frac{1}{2} \sum_{\mathcal{C}} f(\mathcal{C}) X_{\mathcal{C}} } \ket{0} ,
\end{equation}
where $\beta \equiv - \log \tanh K = 2 \artanh[p/(1-p)]$. We see that $\ket{\phi(p)}$ is nearly identical to the deformed toric code state $\ket{\psi(\beta)}$, except that the deformation operator is $e^{\frac{1}{2} \hat{W}}$ instead of $e^{\hat{W}}$. By an identical construction as for $\ket{\psi(\beta)}$, we obtain the following parent Hamiltonian for $\ket{\phi(p)}$:
\begin{equation}
    H_{\text{dec}}(p) = \sum_e \qty( e^{- \sum_{\mathcal{C} \in \mathcal{S}_e} f(\mathcal{C}) X_{\mathcal{C}}} - Z_e ) ,
\end{equation}
which differs from the parent Hamiltonian $H(\beta)$ only by a factor of two inside the exponential term.

Although we will not explicitly prove that $H_{\text{dec}}(p)$ is gapped for sufficiently large $p$, it is clear that an essentially identical proof to the one of Appendices~\ref{app:locality_proof} and ~\ref{sec:proof_of_eq13} will lead to an analogous result: namely, $H_{\text{dec}}(p)$ can be shown to be local in the sense of Eqs.~\eqref{eq:hastsings_1}, \eqref{eq:hastings_2}, and \eqref{eq:hastings_3}, and therefore Theorem~\ref{theorem:BHM} proves that $H_{\text{dec}}(p)$ is gapped and exhibits the unique ground state $\ket{\phi(p)}$ for sufficiently large $p$.

\subsection{More General Decohered Toric Code States}
Thus far, we have studied the ++ state $\ket{\TC_{++}}$ under strong Pauli-$Z$ decoherence. It is interesting to consider more general initial toric code states, which can be written in the form
\begin{equation}
    \ket{\TC_{\Xi}} = \sum_{\ell} \Xi_{\ell} \overline{Z}_{\ell} \ket{\TC_{++}}, 
\end{equation}
where $\overline{Z}_{\ell} \equiv Z_{\mathcal{L}_{\ell}}$ are a collection of four logical-$Z$ operators (for example: $\overline{Z}_0 = 1$, $\overline{Z}_1 = Z_{\tilde{\mathcal{L}_x}}$ for a non-contractible $x$-loop $\tilde{\mathcal{L}}_x$, etc), and $\Xi_{\ell}$ are four complex numbers. Since each $\overline{Z}_{\ell}$ commutes with the Pauli-$Z$ decoherence channel, we can immediately write down the corresponding decohered state:
\begin{equation}
\label{eq:decohered_tc_generic}
    \begin{split}
        \rho_p \equiv \sum_{\ell \ell'} \Xi_{\ell} \Xi_{\ell'}^* \, \overline{Z}_{\ell} \, \mathcal{E}(\dyad{\TC_{++}}) \overline{Z}_{\ell'} &\propto \sum_{\ell \ell'} \Xi_{\ell} \Xi_{\ell'}^* \, \overline{Z}_{\ell} \, \qty( \sum_{\mathcal{L}} X_{\mathcal{L}} \dyad{\phi(p)} X_{\mathcal{L}} ) \overline{Z}_{\ell'} \\
        &= \sum_{\mathcal{L}} \dyad{\Xi_{\mathcal{L}}(p)} ,
    \end{split}
\end{equation}
where
\begin{equation}
    \ket{\Xi_{\mathcal{L}}(p)} \equiv \sum_{\ell} \Xi_{\ell} \overline{Z}_{\ell} X_{\mathcal{L}} \ket{\phi(p)} .
\end{equation}
For sufficiently large $p$ and generic complex coefficients $\Xi_{\ell}$, each state $\ket{\Xi_{\mathcal{L}}(p)}$ is once again SRE up to exponentially small corrections. To see this, recall that $\ket{\phi(p)} = e^{\frac{1}{2} \hat{W}} \ket{0}$ is a sum over $Z_e = -1$ loop configurations such that large loops are suppressed exponentially in their system size. Consequently, $\ket{\phi(p)}$ is approximately 1-form symmetric under non-contractible loops---that is, for each nontrivial logical operator $\overline{Z}_{\ell}$,
\begin{equation}
    \overline{Z}_{\ell} \ket{\phi(p)} = \ket{\phi(p)} + e^{-\mathcal{O}(L)} ,
\end{equation}
where the correction term is a vector of exponentially small norm (relative to that of $\ket{\phi(p)}$). Therefore, for generic $\Xi_{\ell}$, we have
\begin{equation}
    \ket{\Xi_{\mathcal{L}}(p)} \propto X_{\mathcal{L}} \ket{\phi(p)} + e^{-\mathcal{O}(L)} .
\end{equation}
As a result, $\ket{\Xi_{\mathcal{L}}(p)}$ is a SRE state, we find that once again that Eq.~\eqref{eq:decohered_tc_generic} constitutes an exact decomposition into SRE states for generic $\Xi_{\ell}$.

On the other hand, for certain special choices of the wavefunction amplitudes $\Sigma_{\ell}$, it is possible to obtain long-range entangled (LRE) states $\ket{\Xi_{\mathcal{L}}(p)}$. For example, consider the state with $(\Xi_0, \Xi_1, \Xi_2, \Xi_3) = (1,-1,0,0)$ and $\mathcal{L} = \emptyset$, given explicitly by
\begin{equation}
\label{eq:perverse_example}
    \ket{(1,-1,0,0)_{\emptyset}(p)} \equiv (1 - \overline{Z}_{\mathcal{L}_x}) \ket{\phi(p)} .
\end{equation}
This state is a sum over all loop configurations with a non-contractible $Z_e = -1$ loop about the $y$ cycle of the torus; it can be regarded as a ``loop-wave'' state, in the terminology of Ref.~\onlinecite{Sahay2025FunkyIsing}. Such states are expected to be LRE, and so Eq.~\eqref{eq:decohered_tc_generic} does not constitute an exact decomposition into SRE states. Nevertheless, since non-contractible loops comprise an $e^{-\mathcal{O}(L)}$ component of $\ket{\phi(p)}$, the norm of states such as Eq.~\eqref{eq:perverse_example} is $e^{-\mathcal{O}(L)}$. In terms of normalized SRE states $\ket{\psi_i^{\text{SRE}}}$ and normalized LRE states $\ket{\psi^{\text{LRE}}_j}$, the full density matrix $\rho_p$ therefore has the general structure
\begin{equation}
    \rho_p = \sum_i c_i \dyad{\psi_i^{\text{SRE}}} + \sum_j d_j \dyad{\psi_i^{\text{LRE}}}, \quad \sum_j d_j = e^{-\mathcal{O}(L)} .
\end{equation}

In other words, although $\rho_p$ is not exactly decomposed into a convex sum of SRE states (at least in this particular decomposition), the total probability of all LRE states to in this decomposition of $\rho_p$ is exponentially small; consequently, in the thermodynamic limit, $\rho_p$ cannot be distinguished from a genuine SRE state by any observable. In this specific sense, we can once again regard $\rho_p$ as an SRE mixed state.

\section{Verifying entanglement bootstrap axioms for the strongly deformed toric code}
In this Appendix, we provide an informal physicists' argument showing that the strongly deformed toric code satisfies the A0 and A1 axioms of entanglement bootstrap \cite{Shi_fusion_2020,Shi_Kim_2021,Kim_Lin_Ranard_Shi_2024} approximately, such that the axioms become exact as we scale up the partition size. 

We quickly recap the axioms through \figref{fig:entanglement_bootstrap_axioms}. In this section, we show that both $\Delta(B,C)$ of A0 and $\Delta(B,C,D)$ of A1 scale as $\sim r_C \exp(-2 \beta \ell)$, where $\ell$ is the width of the annulus and $r_C$ is the radius of the inner circle. To show this, we use the result of Ref.~\onlinecite{CastelnovoChamon2007DefToricCode} (see \cite{FradkinMoore2006} for a field-theoretic version) which states that, for any region $A$, we can write the Von Neumann entanglement entropy of the reduced density matrix $\rho_A = \Tr_{\overline{A}} \comm{\dyad{\psi(\beta)}}/\braket{\psi(\beta)}{\psi(\beta)}$ as
\begin{equation}
    S_A = - \sum_{\LL} \frac{e^{-\beta\abs{\LL}}}{\mathcal{Z}} \log\term{\frac{\mathcal{Z}^{\partial A}_{\LL}}{\mathcal{Z}}}, \qquad \mathcal{Z} \equiv \sum_{\LL} e^{-\beta \abs{\LL}}, \qquad \mathcal{Z}^{\partial A}_{\LL} \equiv \sum_{\substack{\LL': \\ \LL' = \LL \text{ on } \partial A}} e^{-\beta \abs{\LL'}}.
\end{equation}
The quantity $\mathcal{Z}^{\partial A}_{\LL}$ can be regarded as the partition function of an Ising model with its configuration on $\partial A$ constrained to agree with that of $\LL$. This relation maps the von Neumann entanglement entropy of the deformed toric code on a subsystem $A$ to the entropy of the subsystem spin configuration a classical Ising model\footnote{Technically the Ising$^\star$ model, i.e., an Ising model summed over both periodic and antiperiodic boundary conditions, so that domain walls may exhibit non-contractible loops.} at inverse temperature $\beta$, where the latter subsystem lives on the boundary $\partial A$. We refer the reader to \cite{CastelnovoChamon2007DefToricCode} for a derivation of the above result using the replica trick.

\begin{figure}
    \centering
    \includegraphics[width=0.8\linewidth]{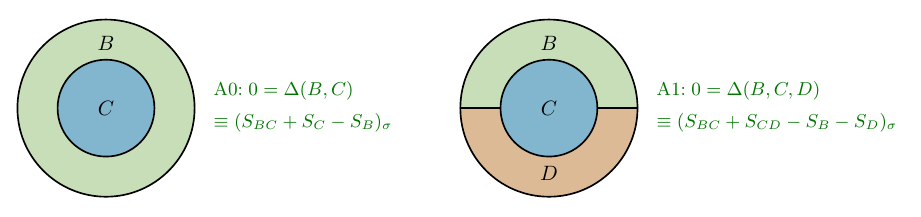}
    \caption{The two axioms of the entanglement bootstrap. Calling $A$ the region outside the annulus in either case, both axioms reduce to having vanishing conditional mutual information $I(A:C \vert B) = 0$. Away from fixed point gapped states, these axioms will only be approximately satisfied, i.e. $I(A:C\vert B) \sim  \exp(-\ell/ \xi)$, where $\ell$ is the width of the annulus and $\xi$ is a partition-independent length scale called the Markov length. In the (infrared) limit where the annulus is scaled up such that $\ell \gg \xi$, the axioms become exact. }
    \label{fig:entanglement_bootstrap_axioms}
\end{figure}

Using $p(\LL) = e^{-\beta \abs{\LL}}/\mathcal{Z}$ as a probability measure over the space of closed loops, we can write
\begin{subequations}
\begin{align}
    \Delta(B,C) &= S_{BC} + S_C - S_B \\
    &= -\mathbb{E}_{\LL}\comm{\log\term{\frac{\mathcal{Z}^{\partial(BC)}_\LL\mathcal{Z}^{\partial C}_\LL}{\mathcal{Z}^{\partial B}_\LL \mathcal{Z}}} }
\end{align}
\end{subequations}
We analyze the argument of the log within a large-$\beta$ (low-temperature in the Ising language) expansion. The highest weight configuration when we take the expectation value $\mathbb{E}_\LL$ over the probability measure $p(\LL)$ is the empty loop configuration $\LL = \{\}$. Correspondingly, the constrained partition functions $\mathcal{Z}^{\partial \mathcal{R}}_\LL$ sum over configurations with no loops cutting across the boundary $\partial \mathcal{R}$. In this sector, we see that all small loop configurations occur an equal number of times in the two `product' partition functions ${\mathcal{Z}^{\partial(BC)}_\LL\mathcal{Z}^{\partial C}_\LL}$ and ${\mathcal{Z}^{\partial B}_\LL \mathcal{Z}}$. The lowest order loop `diagram' that distinguishes the two is a thin connected loop that cuts across both boundaries of region $B$, as such a diagram cannot occur in ${\mathcal{Z}^{\partial(BC)}_\LL, \mathcal{Z}^{\partial C}_\LL}$, or ${\mathcal{Z}^{\partial B}_\LL}$ but can appear in ${ \mathcal{Z}}$. Within a large-$\beta$ (low temperature) expansion, such diagrams contribute as $\sim r_C e^{-2\beta \ell}$ to $\log\term{{\mathcal{Z}^{\partial(BC)}_\LL\mathcal{Z}^{\partial C}_\LL}/{\mathcal{Z}^{\partial B}_\LL \mathcal{Z}}}$, where $\ell$ is the width of the annulus $B$ and $r_C$ is the radius of the inner disk $C$.

For loop configurations $\LL$ consisting of sparse small connected loops, we anticipate that a similar conclusion holds, that  $\log\term{{\mathcal{Z}^{\partial(BC)}_\LL\mathcal{Z}^{\partial C}_\LL}/{\mathcal{Z}^{\partial B}_\LL \mathcal{Z}}} \sim r_C e^{-2\beta \ell}$. Since the low temperature expansion for $\mathbb{E}_\LL$ is dominated by such loop configurations, we argue that for sufficiently large $\beta$, we the deformed toric code satisfies the A0 axiom approximately, given by 
\begin{equation}
    \Delta(B,C) \sim r_C e^{-\ell / \xi},
\end{equation}
where $\xi \approx 1/2\beta$ may be renormalized from the value predicted by this low-temperature expansion by entropic effects that we ignore. Notably, as we uniformly scale up the partition $BC$, $\Delta(B,C)$ vanishes exponentially quickly in linear partition size. 

For the axiom A1, a very similar argument holds. Using the replica trick, we can write 
\begin{subequations}
\begin{align}
    \Delta(B,C,D) &= S_{BC} + S_{CD} - S_B - S_D \\
    &= -\mathbb{E}_{\LL}\comm{\log\term{\frac{\mathcal{Z}^{\partial(BC)}_\LL\mathcal{Z}^{\partial (CD)}_\LL}{\mathcal{Z}^{\partial B}_\LL \mathcal{Z}^{\partial D}_\LL}} }. \label{eq:A1replicaexpression}
\end{align}
\end{subequations}
One can argue that a very similar diagram (a thin connected loop) that cuts across both the inner and outer boundaries of $B$ ($D$) is the smallest ``distinguishing diagram'' that appears in no partition function in (\ref{eq:A1replicaexpression}) except $\mathcal{Z}^{\partial D}_{\LL}$ ($\mathcal{Z}^{\partial B}_{\LL}$) for typical loop configurations $\LL$. Therefore, we expect that, once again 
\begin{equation}
    \Delta(B,C,D) \sim r_C e^{-\ell / \xi}, \qquad \xi \approx 1/2\beta.
\end{equation}
As we increase the partition size compared to $\xi$, $A1$ becomes approximately satisfied.